\documentclass[aps,pra,twocolumn,showpacs,amsmath,amssymb,nofootinbib,
longbibliography,superscriptaddress
]{revtex4-2}

\usepackage[hidelinks]{hyperref}

\usepackage{algorithm}
\usepackage{algpseudocode}

\usepackage{soul}

\usepackage{tikz}
\usetikzlibrary{quantikz2}
\usepackage{physics}

\newcommand\numberthis{\addtocounter{equation}{1}\tag{\theequation}}

\usepackage{graphicx}

\graphicspath{{FIGs/}}

 \usepackage[utf8]{inputenc}

\usepackage{beramono}
\usepackage{listings}

\usepackage{amssymb,amsmath,amsthm}

\newtheorem{theorem}{Theorem}
\theoremstyle{definition}
\newtheorem{definition}{Definition}

\newtheorem{proposition}{Proposition}

\theoremstyle{remark}

\newtheorem{lemma}[theorem]{Lemma}

\lstdefinelanguage{Julia}%
  {morekeywords={abstract,break,case,catch,const,continue,do,else,elseif,%
      end,export,false,for,function,immutable,import,importall,if,in,%
      macro,module,otherwise,quote,return,switch,true,try,type,typealias,%
      using,while},%
   sensitive=true,%
   alsoother={\$},%
   morecomment=[l]\#,%
   morecomment=[n]{\#=}{=\#},%
   morestring=[s]{"}{"},%
   morestring=[m]{'}{'},%
     literate={é}{{\'e}}1
           {è}{{\`e}}1
           {ù}{{\`u}}1
}[keywords,comments,strings]%

\newcommand{\QFT}{\mathrm{QFT}}

\usepackage{orcidlink}

\usepackage{chngcntr}

\begin{document}

\title{Quantum algorithm for the gradient of a logarithm-determinant}

\author{Thomas E.~Baker\orcidlink{0000-0002-3142-0767}}
\email[Please direct correspondence to: ]{bakerte@uvic.ca}

\affiliation{Department of Physics \& Astronomy, University of Victoria, Victoria, British Columbia V8P 5C2, Canada}
\affiliation{Department of Chemistry, University of Victoria, Victoria, British Columbia V8P 5C2, Canada}

\author{Jaimie A.~Greasley\orcidlink{0000-0001-7821-6053}}
\affiliation{Department of Chemistry, University of Victoria, Victoria, British Columbia V8P 5C2, Canada}

\begin{abstract}

The logarithm-determinant is a widely-present operation in many areas of physics and computer science. Derivatives of the logarithm-determinant compute physically relevant quantities in statistical physics models, quantum field theories, as well as the inverses of matrices. A multi-variable version of the quantum gradient algorithm is developed here to evaluate the derivative of the logarithm-determinant. From this, the pseudo-inverse of a sparse-rank input operator may be determined efficiently. Measuring an expectation value of the quantum state--instead of all $N^2$ elements of the input operator--can be accomplished in $O(k/\varepsilon^2)$ time in the idealized case for $k$ relevant eigenvectors of the input matrix with precision $\varepsilon$. A practical implementation of the required operator will likely need $\log_2N$ overhead, giving an overall complexity of $O((k\log_2 N)/\varepsilon^2)$. The method applies widely and converges super-linearly in $k$ when the condition number is high.  The best classical method we are aware of scales as $N$.

Given the same resource assumptions as other algorithms, such that an equal superposition of eigenvectors is available efficiently, the algorithm is evaluated in the practical case as $O(\log_2 N/\varepsilon^2)$. The output is given in $O(1)$ queries of an oracle, which is given explicitly here and only relies on time-evolution operators that can be implemented with arbitrarily small error. The algorithm is envisioned for fully error-corrected quantum computers but may be implementable on near-term machines. We discuss how this algorithm can be used for kernel-based quantum machine-learning. 

\end{abstract}
\maketitle

\section{Introduction}

The logarithm-determinant appears ubiquitously in the literature, notably in quantum and statistical physics, but also in  quantum chemistry and computer science \cite{kardar2007statistical,yuan2018superlinear,yuan2018convergence,li2020randomized,fitzsimons2017bayesian,zhouyin2021understanding}. Finding the logarithm-determinant of a given matrix is known to solve problems in general relativity \cite{weinberg1972gravitation}, quantum field theories \cite{dunne2008functional}, lattice quantum chromodynamics (QCD) \cite{peskin1995quantum}, machine learning \cite{fitzsimons2017bayesian,zhao2019bayesian,mackenzie2014bayesian}, and statistical physics \cite{reif2009fundamentals}. The study of statistical thermodynamics, which spans classical and quantum descriptions of nature, is heavily reliant on the use of the logarithm-determinant of the partition function (written in as a matrix in some basis), but this often requires many basis functions to create an accurate enough representation.

Beyond applications in natural sciences, the logarithm-determinant has a straight-forward connection with the inverse of a matrix.\footnote{Lemmas are used throughout this paper to signify results that are proven by other papers. Theorems are used in this paper to denote results that are given a proof, even if those results were derived somewhere else. Propositions are statements that are easily verified without a detailed proof. Corollaries extend what is already known from a theorem but with some modified detail that is easily verified without restating the proof.}
 \begin{lemma}[Matrix inverse from the derivative of the logarithm-determinant]
The derivative of the logarithm-determinant of $X$ with elements $x_{ij}$ is
\begin{equation}\label{logdetfirst}
\frac{\partial}{\partial x_{ij}}\ln\det(X)=y_{ji}
\end{equation}
where the elements $y_{ji}$ correspond to the transpose elements of the inverse matrix $Y=X^{-1}$. 
\end{lemma}



Classical methods for the logarithm-determinant take a square matrix and decompose it using an LU decomposition. From the resulting form of a lower- (L) and upper-triangular (U) matrix, the determinant can be computed with a cost of $O(N^3)$ for an input $N\times N$ matrix \cite{press1992numerical}.  

If the inverse exists and the spectrum of singular values decays sufficiently, the form given for $X$ can be restricted to a sum over relevant eigenvalues which is far less than the total Hilbert space size, provided that the eigenvectors of $X$ that are retained allow for a sufficiently accurate description of the pseudo-inverse. This can be accomplished by a sparse-rank factorization of the input matrix with a Lanczos algorithm \cite{lanczos1950iteration,baker2021lanczos,bakerPRA24}, but these methods often suffer from stability issues on the classical computer \cite{wang2015improved,eberly1997randomized,rokhlin2010randomized,halko2011algorithm,gu2015subspace,tropp2022randomized,baglama2006restarted,kim1988structural,tropp2022randomized}.

The goal of quantum computing is to evaluate operations like these in a more efficient manner and so we pursue an algorithm that solves this problem faster than the classical computer. The idea is to use the operations on the quantum computer to compute a Hilbert space size of $2^n=N$ for $n$ qubits. The exponential compression of data onto $n$ qubits instead of $N$ objects in the entire Hilbert space is one foundational advantage of many quantum algorithms \cite{aaronson2011computational}. Note that computation of the logarithm-determinant is not the same as the determinant nor the permanent of a matrix \cite{aaronson2011linear,grier2016new}.

This motivates us to look for a quantum algorithm that can stably compute the logarithm-determinant with fewer resources than the classical approach. Such a solution would have wide-ranging implications on the ability to compute quantities in many areas. Some review with mathematics is provided in Appendix~\ref{logdet_phys}.

The operations required to compute Eq.~\eqref{logdetfirst} are efficient on the quantum computer, as we show in this paper. The logarithm-determinant for sparse-rank matrices is no more costly than the quantum phase estimation (QPE)  algorithm in the case where the input unitary can be efficiently provided \cite{kitaev1995quantum,nielsen2010quantum}. 

For evaluating the derivative in Eq.~\eqref{logdetfirst}, we use a quantum gradient method--which is based on the ideas of the quantum gradient algorithm from Ref.~\onlinecite{jordan2005fast}--that computes the multi-variable series expansion of an input function. Here, the oracle required is given explicitly by a QPE and is the dominant cost of the algorithm. In the practical case, it is shown that the full scaling algorithm is poly-logarithmic to determine the expectation value of the pseudo-inverse matrix, $\langle Y\rangle$.

A necessary starting point for the current algorithm is the generation of a set of the most relevant eigenvectors with sufficient accuracy. We reference the Lanczos method to prepare this set of eigenstates for the input operator \cite{baker2021lanczos,bakerPRA24}. Given a set of eigenstates, the subsequent step of computing the quantum gradient of the logarithm-determinant (QGLD) will allow for the accurate pseudo-inverse of a given sparse-rank operator, which is relevant for quantum machine learning. This idea is extended to note that for an equal superposition of eigenstates, one could generate all elements at once, which makes the evaluation of a non-sparse-rank operator possible provided the correct initial state.

\section{Establishing Lemma~\ref{logdetfirst}}

The proof for Lemma~\ref{logdetfirst} has been long known in the literature \cite{kubota1994matrix}, but the proof utilizes several concepts that are useful for analysis of the results and discussion of its validity later on. We will restrict our attention to Hermitian matrices $X$ for simplicity.

\begin{proof}[Proof of Lemma~\ref{logdetfirst}.]

Consider each element of $X$ (denoted $x_{ij}$) as existing on a continuous domain of real numbers. Thus, we can define a {\it directional derivative} as
\begin{equation}
\nabla_\Delta X=\Delta\cdot\nabla X
\end{equation}
where the nabla operator acting on $X$ supposes there is a continuous parameter 

Choose a Hermitian matrix $X$ and consider that the derivative of the directional determinant of $X$ is related by
\begin{equation}\label{grad_det}
\nabla_{\Delta}\mathrm{det}(X)=\underset{\epsilon\rightarrow0}\lim\left(\frac{\mathrm{det}(X+\epsilon\Delta)-\mathrm{det}(X)}\epsilon \right)
\end{equation}
where $\nabla_\Delta$ is called the directional derivative over the perturbations contained in a matrix $\Delta$ of the matrix $X$. The matrix $\Delta$ corresponds to all the perturbations of the entries of $x_{ij}$. 

The form of Eq.~\eqref{grad_det} can be rewritten as
\begin{align}
\nabla_{\Delta}\mathrm{det}(X)&=\mathrm{det}(X)\underset{\epsilon\rightarrow0}\lim\left(\frac{\mathrm{det}(I+\epsilon X^{-1}\Delta)-1}\epsilon \right)\\
&=\mathrm{det}(X)\nabla_{(X^{-1}\Delta)}\mathrm{det}(I)
\end{align}
The direction derivative $\nabla_{(X^{-1}\Delta)}\mathrm{det}(I)$ has a special relationship with the trace. If we exhaustively examine the components of $ \mathrm{det}(I+\epsilon X^{-1}\Delta)$, then we would find a pattern. The $\epsilon=0$ term is the identity. All quadratic and higher powers of $\epsilon$ will vanish in the limit as $\epsilon\rightarrow0$. This leaves only the $O(\epsilon)$ remaining.  These terms are noticed to be synonymous with the trace of the function. Thus, here we can write
\begin{equation}
\nabla_{(X^{-1}\Delta)}\mathrm{det}(I)=\mathrm{Tr}(X^{-1}\Delta)
\end{equation}
and there is one further reduction. So far we have not explained the meaning of $\Delta$. We suppose that this matrix is composed of the differential elements applied onto each of $x_{ij}$. When we take the derivative of the determinant (not the directional derivative) by some variable, then we would have
\begin{equation}
d(\mathrm{det}(X))=\mathrm{det}(X)\mathrm{Tr}(X^{-1}d\Delta)
\end{equation}
since only $\Delta$ contains the derivative of the matrix, or that $d\Delta=dX$ with elements $dx_{ij}$. Taking the derivative with respect to $\frac d{dx_{k\ell}}$, we have a delta function, $\frac {dx_{ij}}{dx_{k\ell}}$. 

Returning now to the gradient of the logarithm-determinant instead of the determinant itself, we have
\begin{equation}\label{jacobifull}
\nabla_X\ln\mathrm{det}(X)=(X^{-1})^T
\end{equation}
where by the basic properties of the trace, one can note that $\mathrm{Tr}(AB)=\sum_{ij}A_{ij}B_{ji}\equiv\langle A^T,B\rangle$. Breaking Eq.~\eqref{jacobifull} into individual components so a gradient over $x_{ij}$ recovers the statement from Lemma~\ref{logdetfirst}.
%
%
\end{proof}

The proof here is an extension of Jacobi's formula which expresses the derivative of the determinant of a matrix and has extensions in general relativity \cite{weinberg1972gravitation}.

We note in this formulation the direction derivative defined is taken to be taken on a function that maps the matrix to a scalar value. That function here is the determinant. 

Having established the validity of Lemma~\ref{logdetfirst}, we ask what a quantum extension of this quantity would look like and how efficient it should be. 

\section{Quantum gradient of the logarithm-determinant}

Obtaining the derivative of a function on the quantum computer is known to be obtainable in $O(1)$ oracle queries \cite{jordan2005fast}. The only hurdle is the efficient implementation of the oracle.

We will formulate the gradient algorithm for eigenstates of a matrix operator. This was not an assumption in the original implementation \cite{jordan2005fast}, but it will be useful to understand this to extend the algorithm to a multi-variable series expansion. In contrast, Ref.~\onlinecite{jordan2005fast} assumed an input function $f$ for a smooth, real function. That function in our work will be extended to a matrix with few conditions on it except that the pseudo-inverse exists.

The computation of the logarithm-determinant can be found through eigenvalues and their derivatives. To avoid potential issues with degeneracies, we will assume a non-degenerate matrix is input.
\begin{theorem}[Pseudo-inverse operator from eigenvalues]
Given the $k$ most relevant eigenvalues and eigenvectors and their derivatives, the derivative of the logarithm-determinant can be expressed as
\begin{equation}\label{logdetequiv}
y_{ji}=\frac{\partial}{\partial x_{ij}}\ln\det(X)=\sum_{p=1}^k\frac{\delta E_p(x_{ij})}{E_p}
\end{equation}
where $E_p$ is the eigenvalue of an operator $X$ indexed by $p$ and ordered from the $k$ most relevant eigenstates, $x_{ij}$ is the element of $X$ that a gradient is taken over, and $i$ and $j$ index the elements of the operator $X$.
\end{theorem}

\begin{proof}
Consider the reduction of the problem in Eq.~\eqref{logdetfirst} by first representing the operator $X$ in the eigenbasis. In this case,
\begin{equation}\label{sparserankX}
X=\sum_{p=1}^kE_p|p\rangle\langle p|
\end{equation}
for eigenvalues $E_p$, eigenvectors $|p\rangle$, and rank-$k$ (where $k\in\{1,\ldots,N\}$). The determinant of such a form is 
\begin{equation}
\mathrm{det}(X)=\prod_{p=1}^kE_p
\end{equation}
which the same as that of a diagonal matrix. Since this form of the determinant is approximated, the resulting inverse will be referred to as a pseudo-inverse for this reason. 

The key identification of using this form is that the derivative of the logarithm gives
\begin{align}
\frac{\partial}{\partial x_{ij}}\ln\det(X)&=\frac{\sum_w\left(\left(\prod_{p\neq w}E_p\right)\frac\partial{\partial{x_{ij}}}E_w\right)}{\prod_p E_p}\\
&=\sum_w\frac{\partial_{x_{ij}}E_w}{E_w}\equiv\sum_w\frac{\delta E_w(x_{ij})}{E_w}\label{equivlogdet}
\end{align}
which is given by the product rule. For the sake of brevity, we denote $\partial_{x_{ij}}E_w=\partial E_w/\partial {x_{ij}}\equiv\delta E_w(x_{ij})$.
\end{proof}

Some examples of this form are contained in Appendix~\ref{classicalchecks}. A similar proof follows from the trace-logarithm form that is equivalent to the logarithm-determinant (see Appendix~\ref{logdet_phys}).

This remarkably simple form relates the eigenvalues and their derivatives to the derivative of the logarithm-determinant, and therefore the matrix inverse. Note that this form is valid even for negative eigenvalues (see Appendix~\ref{neg_eigvals}), and note that the convergence is best when the number of eigenvectors required is small and much less than the full Hilbert space size, $k\lll N$.

We will discuss existing methods to generate a suitable subspace of eigenvectors ({\it i.e.}, wavefunction preparation) and then describe a method to find the derivatives of the eigenvalues. The resulting algorithm will determine expectation values of a suitable approximation to $Y$ in poly-logarithmic time provided that we measure the expectation value $\langle Y\rangle$ directly on the quantum computer instead of trying to discover all $O(N^2)$ elements. The recovery of all $N^2$ components in anything less than $O(N^2)$ time is not promised in this paper and would run contrary to expectation in that an exponential amount of data should take an exponential amount of time to query \cite{preskill2018quantum}. This is similar to how the output in a quantum Fourier transform (QFT) are considered \cite{nielsen2010quantum}.

\begin{figure}
	\begin{algorithm}[H]
		\caption{Random quantum block Lanczos (RQBL)}
		\label{RQBL}
		\begin{algorithmic}[1]
			\State Choose a state $\Psi_0=\Psi$ and $b$ excitations of operator $X$, denoted as $\boldsymbol{\Psi}_0$  \Comment{From Lanczos \cite{baker2021lanczos} or otherwise}
			\State $p \gets 0$
			\State $\mathbf{B}_{-1} \gets \mathbf0$
			\State $\mathbf{A}_{0} \gets \boldsymbol{\Psi}_0^\dagger X\boldsymbol{\Psi}_0$ \Comment{Eq.~\eqref{An}}
			\While{$p < k$}
			\State Compute recursion for $p$th step \Comment{Eq.~\eqref{blocklanczos}}
			\State Measure $\mathbf{A}_{p+1}$ (stored classically) \Comment{Eq.~\eqref{An}}
			\State Measure $\mathbf{B}_{p+1}$ (stored classically) \Comment{Eq.~\eqref{Bn}}
			\State Create $\mathbf{G}_p$ such that $\boldsymbol{\Psi}_p \gets \mathbf{G}_p\boldsymbol{\Psi}_0$ from all of $\{\mathbf{A}_p\}$, $\{\mathbf{B}_p\}$
			\State $p \gets p + 1$
			\EndWhile
			\State Create $S^{(p)}$ \Comment{Eq.~\eqref{rank_k_lanczos}}
			\State $\gamma^{(p)} \gets \mathrm{diag}(S^{(p)})$
			\State \textbf{return} Eigenvectors $|\psi_p\rangle$ \Comment{Ref.~\onlinecite{bakerPRA24}}
		\end{algorithmic}
	\end{algorithm}
\end{figure}

\subsection{Determination of eigensolutions}

At the start of this algorithm, $k$ eigenvalues are required from an input $X$. We provide in Appendix~\ref{R-QBLR} how to do this from a block-Lanczos technique \cite{baker2021lanczos,bakerPRA24} and summarized in Box~\ref{RQBL}. In principle, block-Lanczos is not required and other methods may be used in place of it if they achieve sufficient accuracy of the eigenvectors. For the current paper, this technique is applied to demonstrate that the wavefunction preparation is not exponentially large \cite{tang2021quantum}.

\subsection{Warm-up: Quantum gradient algorithm on eigenstates}

The general structure of a quantum gradient algorithm (QGA)  \cite{jordan2005fast} to generate the derivative of a smooth function $\theta(x)$ for some real input $x$ and function $\theta$. The circuit for the QGA is presented in Fig.~\ref{QGA_alt}.

Our formulation here emphasizes the use of eigenvectors in the operations of the QGA, so that we can extend the concept later on.

\begin{figure}[b]
\begin{quantikz}[column sep=8, wire types = {n, b}]
&&&&& \lstick{$\ket{0}^{\otimes m}$}&\qwbundle[Strike Height= 0.2]{}\setwiretype{b}&\gate{H^{\otimes m}} \slice{}&\phase{\epsilon}\slice{}&\gate{\mathrm{QFT}^{\dagger}}\slice{}&\meter{} \\
\lstick{$\ket{0}^{\otimes n}$}& \qwbundle[Strike Height= 0.2]{} &\gate{\Gamma} \slice{} && \push{\ket{\Psi}}&&&&\gate{\theta(\epsilon)} \wire[u][1]{q} &&
\end{quantikz}
\caption{Circuit for evaluating the gradient of a function $\theta$.
\label{QGA_alt}
}
\end{figure}

\begin{theorem}[Quantum gradient algorithm on eigenstates] \label{QGArephrased}
Obtaining a gradient of an eigenvalue from the quantum computer for a eigenvector $\Psi$ costs only $O(1)$ query time to an oracle and two registers of $n=\log_2N$ and $m=\log_2M$ qubits respectively.
\end{theorem}

\begin{proof}
Given two registers
\begin{equation}
|0\rangle^{\otimes m}|0\rangle^{\otimes n}
\end{equation}
a Hadamard gate, $H$, is applied onto the first register as for each of $n$ qubits.

The state of the circuit is written as: 
\begin{align}
\left(H|0\rangle\right)^{\otimes m}|0\rangle^{\otimes n}
&=\left(\frac1{\sqrt{M}}\sum_{\epsilon = 0}^{M-1}|\epsilon\rangle\right)|0\rangle^{\otimes n}
\end{align}
where $\epsilon$ indexes all bit-strings.

An eigenstate is created on $n$-qubits,
\begin{align}
&\left(\frac{1}{\sqrt{M}} \sum_{\epsilon} |\epsilon\rangle \right) \Big(\Gamma|0\rangle^{\otimes n}\Big)=\left(\frac{1}{\sqrt{M}} \sum_{\epsilon} |\epsilon\rangle \right)|\Psi\rangle
\end{align} 
where $\Gamma$ creates the eigenstate, $\Psi$.

We then perform the phase rotation gates, $\mathcal{U}_c$, from the QPE on the eigenvector. The QPE would require a modification of the gate $\exp(i\hat X)\rightarrow\exp(iX(\epsilon))$ such that it produces a phase
\begin{equation}
\left( \frac{1}{\sqrt{M}}\sum_{\epsilon} e^{i2\pi\frac{M}{WL}\theta\left(\frac{L}{M}\left(\epsilon-\frac{M}{2}\right)\right)} |\epsilon\rangle \right)|\Psi\rangle
\end{equation}
where $W$ is the maximum expected value for the gradient, and $L$ is a parameter representing a small linear region for the domain of which the function is being evaluated. In this example, the QPE would be aided if the operator $X(\epsilon)\approx X$, implying a small perturbation in some parameter. We regard this an oracle query, the only such application of this oracle.

If we choose $L$ sufficiently small, then we can approximate
\begin{equation}
\theta\left(\frac LM\left(\epsilon-\frac{M}2\right)\right)\approx \theta({0})+\frac LM\left(\epsilon-\frac{{M}}2\right)\cdot{\nabla}\theta
\end{equation}
which shows now the gradient $\nabla\theta$. The leading term $\theta(0)$ contributes only a trivial, global phase. Isolating $\mathbf{\epsilon}$ simplifies to
\begin{equation}
\frac{1}{\sqrt{M}}\sum_{\epsilon=0}^{M-1}e^{i\frac{2\pi }{W}\epsilon\cdot{\nabla}\theta}|{\epsilon}\rangle
\end{equation}
which gives the gradient after an inverse QFT.
\end{proof}

In the original formulation, Ref.~\onlinecite{jordan2005fast} used planewaves in superposition for the eigenvectors. Then, Ref. \onlinecite{jordan2005fast} sets the actual value added to the second register as $\frac{MN}{WL} f(x)$ rounded to the nearest integer. Thus, the $\Gamma$ operator consists of an operation to flip the $m$-bit register into a set of 1 states (instead of 0) and then apply the QFT. A modular addition was used for the phase kickback step instead of the phase rotation gates from the QPE in the proof to obtain a phase of the correct form. We show this explicitly in Appendix~\ref{QGAapx}. Regardless, the QGA strongly resembles the QPE and is a special case of the general construction above.

What Thm.~\ref{QGArephrased} emphasizes is that the QGA is formulated on a phase kick-back trick on an eigenvector. In Ref.~\onlinecite{jordan2005fast}'s case, the eigenvector is plane-waves, which are eigenvectors of the number operator. The modular addition leaves this state unperturbed. In this new generalization, Thm.~\ref{QGArephrased}, the eigenvector is allowed to be of any operator. This will be useful in extending the concept to the gradient of matrix elements.

\subsection{Quantum gradient phase estimation algorithm}\label{oracle_section}

Now we present a quantum algorithm for deriving the gradients of the eigenvalues for the eigenstates of a matrix. In effect, we map the problem to a quantum phase estimation and use a different operator form.

Oracle-based methods can often produce large overheads to implement the oracle. We show in the following that there is an efficient implementation. 

The necessary gate that is applied here is imagined to be similar to a quantum phase estimation such that the register conveying $f(\epsilon)$ is an eigenphase of that operator in the limit where $L$ is small. The use of phase estimation here is another departure from Ref.~\onlinecite{jordan2005fast}.

We now show what form the operator must take for the logarithm-determinant as an extension of the QGA.

The oracle query for the QGA--which we refer to as quantum gradient phase estimation (QGPE)--that is shown to be no more costly than QPE--controlled on a parameter ${\epsilon}$ for an input operator $X$. The circuit for this algorithm is identical to Fig.~\ref{QGA_alt} but substituting the operator $\theta(\epsilon)\rightarrow \mathcal{U}(\epsilon)$ which is elaborated on further below.

\begin{theorem}[Oracle query as a quantum gradient phase estimation]

The oracle query for the quantum gradient algorithm can be phrased as a quantum phase estimation controlled on an equal superposition of bit strings. 

\end{theorem}

\begin{proof}
Since the derivative of the logarithm-determinant was reduced to the gradient of the eigenvalues in Eq.~\eqref{equivlogdet}, the intention is now to apply a phase rotation gate onto the wavefunction--as is performed in a QPE--which will give the eigenvalues. In order to account for the variation $\epsilon$, the phase rotation gates in the QPE must be controlled on those registers. This will result in the energies with a variation in $\epsilon$.

The variation is applied of the form $\epsilon\Delta(i,j)$ where $\Delta(i,j)$ has 1 on all elements of the input matrix that are to be verified in the derivative $\partial_{x_{ij}}$ which can be taken to be a set of Kronecker delta functions. The remaining values are all 0. The form of $\Delta(i,j)$ can be taken to be symmetric so that the determinant of the input matrix $X$ is commonly defined, although one in principle can extend the concept if necessary.
 
The form of the phase rotation gates in the QPE are taken to have the form
\begin{align}
\mathcal{U}(\epsilon)&=\exp\left({\frac {2\pi iN}{WL}\hat H(\epsilon)}\right)\\
&=\exp\left({\frac{2\pi iN}{WL}\left(X+\frac{L\epsilon}N\Delta(i,j)\right)}\right)\label{timeEvolve}
\end{align}
where $\hat H$ can be regarded as a Hamiltonian with time step $2\pi N/(WL)$ with $W$ the maximum value of the gradient. We concede that the exact form of Eq.~\eqref{timeEvolve} may be complicated to execute on the quantum computer, but this form directly extends the QGA. We will quote the computational complexity for this unitary as if we had an efficient form.

In the following, one can imagine taking an expectation value of Eq.~\eqref{timeEvolve} with respect to $|p\rangle$, $\langle p|e^{i2\pi\frac N{WL}\hat H(\epsilon)}|p\rangle$. For the $\Delta(i,j)$ operator here, we consider a form $\Delta(i,j)$ such that a choice number of elements are non-zero. For simplicity, consider only one (and its transpose position to keep the operator of a suitable form to take a unitary time evolution).

The result of the QPE applied onto $|p\rangle$ is to generate a phase $\exp(i2\pi N\theta^{(p)}(\epsilon)/L)$ where the eigenvalue is given as
\begin{equation}\label{qgaphase}
\theta^{(p)}(\epsilon)=\theta_0^{(p)}+\delta\theta^{(p)}+\ldots=\theta_0^{(p)}+\nabla_{ij}\theta^{(p)}\frac{L\epsilon}N+O(L^2)+\ldots
\end{equation}
and applies for any eigenvalue indexed by $p$. For small enough $L$, the condition holds as before. The end result is that a global phase $\exp(i2\pi\theta_0^{(p)})$ present may be disregarded in the rest of the computation. After taking a QFT, the eigenvalue reduces to a form similar to Eq.~\eqref{gradientoutcome},
\begin{equation}
\left|\frac NW\nabla_{ij}\theta^{(p)}\right\rangle
\end{equation}
where the number of qubits $m$ allowed for the $\epsilon$ deviation determines how precise the answer can be.

The scaling for this step is given entirely by the number of qubits required to represent $\varepsilon$ where $0<\varepsilon<1$. For the case of qubits, we can impose $m=\log_2\frac1\varepsilon$ and use the known resource estimates from QPE.
\end{proof}

Implementation details for the QGPE is contained in Appendix~\ref{qgaimplement}.

There is an additional question of how accurate the gradient is. We demonstrate that for small $L$, the time-evolution operator is essentially error-free.

\begin{theorem}[Small $L$ limit for error-free computation]
\label{smallLthm}

The error of the algorithm in the variable $L$ can be reduced to zero.

\end{theorem}

\begin{proof}
Given the argument of the exponential operator
\begin{equation}\label{Hexpanded}
\hat H_{ij}=(X)_{ij}+\frac {L\epsilon}N\Delta(i,j)
\end{equation}
we can determine that the expansion of the perturbing term $L\epsilon\Delta(i,j)/N$ follows as a series expansion of the term as seen in Eq.~\eqref{qgaphase} (see Appendix~\ref{multiTaylor} for the full form of a multi-variable series). Higher order terms in Eq.~\eqref{Hexpanded} therefore appear as $O(L^2)$ corrections. The prefactor from Eq.~\eqref{timeEvolve} of $N/(WL)$ will cancel out the extraneous factor giving an overall term of $O(L)$. As $L$ is decreased, the error will vanish. In total, higher order terms will go as $O(L^{o-1})$ for order $o$ and so vanish faster than previous orders. 

In the series expansion, we have a leading term that goes as $O(\frac1L)$, but this first term counts as a global phase and therefore does not affect the final result. We then have the first gradient term as $O(1)$ which is the desired result. The rest of the terms go as, at most, $O(L)$ or larger which can be made small enough not to affect the final result.
\end{proof}

The actual evaluation of the series expansion here would be difficult to do in practice, so it may be useful to think of the term $e^{L\epsilon\Delta/N}$ as modifying the eigenvector $|p\rangle$. Reducing $L$ will guarantee less change in the eigenvector and there is no restriction for how small $L$ can be.

This then guarantees that for sufficiently small $L$ that the error will vanish. One could ponder a similar justification through the Baker-Campbell-Hausdorff formula for the decomposition of an exponential, but it is less straightforward to determine. We will note that practical implementation may limit how small $L$ can be, but theoretically in a perfect implementation, this parameter can be arbitrarily small.

Upon evaluation of the time-evolution operator, we therefore have a state that is nearly an eigenstate and the variation to the new Hamiltonian is a small perturbation. Estimates for real-time evolution that suggest lengthy time evaluations are not applicable in this particular use case because the time-step is infinitesimally small, allowing for arbitrary time-steps without error \cite{poulin2014trotter}. 

Appendix~\ref{qgaimplement} contains some examples.

\subsection{Expectation values of the psuedo-inverse}\label{expvalues}

The goal in this section is to see if the computation of the pseudo-inverse operator can be performed efficiently on the quantum computer without having to measure $O(N^2)$ elements as the previous, single-shot algorithm suggests ({\it i.e.}, all possible combinations of $i$ and $j$ must be obtained in order to find all elements of the pseudo-inverse matrix). In total, we want to evaluate
\begin{equation}\label{measY}
\langle Y \rangle=\langle\Phi|Y|\Phi\rangle=\sum_{ij}\Phi^*_iy_{ij}\Phi_j
\end{equation}
which scales as $O(N^2)$ on the classical computer as a sub-leading cost to the determination of the inverse at $O(N^3)$. In this section, our interest is on how to evaluate the double sum over $i$ and $j$ in poly-logarithmic time on the quantum computer to do both the pseudo-inverse (not measured in full) and expectation value.

\begin{figure}
	\begin{algorithm}[H]
		\caption{Quantum gradient of the logarithm-determinant (QGLD)}
		\label{QGLD}
		\begin{algorithmic}[1]
			\State Generate $k$ eigenvectors $\{|p\rangle\}$ of an operator $X$ and their eigenvalues $\{E_p\}$ \Comment{{\it e.g.}, block Lanczos \cite{baker2021lanczos,bakerPRA24}}
			\State Prepare operator with elements $x_{ij}+L\epsilon\Delta(i,j)/N$ controlled on auxiliary qubits $|\epsilon\rangle^{\otimes m}$
			\State Choose $L$ to be sufficiently small
			\State Prepare a state $\Phi$ for the expectation value
			\State $p \gets 1$
			\While{$p < k$}
			\State Prepare deviations $|\epsilon\rangle^{\otimes m}$  \Comment{Eq.~\eqref{deviations}}
			\State QGPE onto wavefunction $|p\rangle$ controlled on $|\epsilon\rangle^{\otimes m}$ to produce the phase $e^{2\pi\epsilon i\delta E_p(x_{ij})}$ \Comment{Eq.~\eqref{timeEvolve}}
			\State $|\delta E_p(x_{ij})\rangle \overset{\mathrm{QFT}}\gets e^{2\pi\epsilon i\delta E_p(x_{ij})}|\epsilon\rangle^{\otimes m}$ and measure or replacing $\delta E(x_{ij})\rightarrow\sum_{ij}\delta E_p(x_{ij})$ for all perturbed entries
			\State $p \gets p +1$
			\EndWhile
			\State \textbf{return} $\langle Y\rangle\gets \sum_{p=1}^k\langle Y^{(p)}\rangle$
		\end{algorithmic}
	\end{algorithm}
\end{figure}

The broad outline of the algorithm is given in Box~\ref{QGLD}.

If $\Phi$ is given as a classical wavefunction, we can immediately implement the coefficients into the QGPE by modifying the perturbing operator $\Delta$ as in the following.

\begin{theorem}[Summation of all derivatives]
\label{sumderivthm}

Assuming a form of 
\begin{equation}\label{outerPhiProd}
\Delta=\frac{L\epsilon}N\boldsymbol{\Phi}^\dagger\cdot\boldsymbol{\Phi}=\frac{L\epsilon}N\left(\begin{array}{cccccccc}
\Phi_1^*\Phi_1 & \Phi_1^*\Phi_2 & \Phi_1^*\Phi_3 & \cdots\\
\Phi_2^*\Phi_1 & \Phi_2^*\Phi_2 & \Phi_2^*\Phi_3 & \ddots\\
\Phi_3^*\Phi_1 & \Phi_3^*\Phi_2 & \Phi_3^*\Phi_3 & \ddots\\
\vdots & \ddots & \ddots & \ddots
\end{array}\right)
\end{equation}
generates the sum of all derivatives. The weights of the derivatives are modified here by the components $\Phi^*_i\Phi_j$, the probability amplitudes of the wavefunction $\Phi$.
\end{theorem}

\begin{proof}
For the sake of simplicity, we can think of the perturbation $\Delta=L\epsilon\mathcal{I}/N$ and define
\begin{equation}
\mathcal{I}\equiv\left(\begin{array}{ccccccc}
1 & 1 & 1 & 1 & 1 & \ldots & 1\\
1 & 1 & 1 & 1 & 1 & \ldots & 1\\
1 & 1 & 1 & 1 & 1 & \ddots & 1\\
\vdots & \vdots & \vdots & \vdots & \ddots & \ddots & \vdots\\
1 & 1 & 1 & 1 & 1 & \ldots & 1
\end{array}\right)
\end{equation}
as a matrix of all ones in every entry. It is natural to wonder why this choice would lead to a sum over all elements $\sum_{ij}\delta E_p(x_{ij})$ for a given eigenvalue $p$, but this case can be analogized with a multi-variable series expansion. Each element of $\hat H_0$ indexed by $i$ and $j$.
\begin{align}
X+\frac{L\epsilon}N\mathcal{I}&=X+\frac{L\epsilon}{N}\left(\sum_i\delta_{ii}+\sum_{i< j}\left(\delta_{ij}+\delta_{ji}\right)\right)\label{deltasum}\\
&=X+\frac{L\epsilon}{2N}\sum_{ij}|i\rangle\langle j|
\end{align}
where the second sum in Eq.~\eqref{deltasum} is taken over each pair of $(i,j)$ only once. 

We now imagine that we want to series expand in several factors of $i$ and $j$ simultaneously. To first order, we can imagine that we can take the multi-variable series expansion of the expression for all possible $i$ and $j$ as the expression
\begin{equation}
\langle p|X+\frac{L\epsilon}N\mathcal{I}|p\rangle=E_p+\langle p|\sum_{ij}\frac{L\epsilon}{2N}\Big(|i\rangle\langle j|\Big)|p\rangle
\end{equation}
where we notice that a single value of $E_p$ is present just as in the 0th order of the multi-variable series expansion followed by a term that when expanded in a series will yield a derivative with prefactor $L$. The same analysis follows as from Thm.~\ref{smallLthm}. The first order derivative then appears as before for small enough $L$. The maximum element of the commutator of $X$ and $\Delta$ will be the leading error in the expression and that the factor $L$ must suppress this factor in order to have low error. 

In total, with the operator $\mathcal{I}$, we have obtained the sum of the derivative of all eigenvalues $\delta E_p(x_{ij})$ which can be viewed as evaluating $\langle Y\rangle$ with a vector $\mathbf{\Phi}^\dagger=(1,1,1,1,\ldots,1)$ and adding in eigenvalue elements as we demonstrate next.

Adding in arbitrary values of the classical wavefunction to each element corresponding to $i$ and $j$ gives
\begin{align}
\sum_{ij}\Phi^*_i\Big(\delta E_p(x_{ij})\Big)\Phi_j&=E_p\sum_{ij}\Phi^*_iy^{(p)}_{ji}\Phi_j\\
&\equiv E_p\langle Y^{(p)}\rangle
\label{Yp}
\end{align}
which corresponds to Eq.~\eqref{measY} with the elements $y_{ji}^{(p)}$ corresponding to the $p$th contribution to the pseudo-inverse element from Eq.~\eqref{logdetfirst}. We then recover the statement of the proof once divided by $E_p$ as in Eq.~\eqref{Yp} and summed, $\langle Y\rangle=\sum_p\langle Y^{(p)}\rangle$. This is the expectation value of the pseudo-inverse in $O(1)$ applications of QGPE.
\end{proof}

Note that the sum of these values as presented here will be for a given eigenvalue $p$. All $k$ such relevant eigenvalues must be summed to obtain the full result.

We note that if there is difficulty in decomposing the input time-evolution gate, then the operator could be chunked over several runs since the output gradients can all be summed together.

\subsection{Complexity scaling}

The overall computational scaling of the algorithm with the classical wavefunction is $O(k/\varepsilon^2)$ \cite{nielsen2010quantum,coppersmith2002approximate} because we require $k$ eigenvectors at $\varepsilon$ precision and a QFT which scales as $O(1/\varepsilon^2)$ \cite{nielsen2010quantum,coppersmith2002approximate}.

If given a quantum wavefunction, a series of measurements could be performed to obtain the classical coefficients \cite{aaronson2018shadow} or from QAM-sampling \cite{marriott2005quantum,temme2011quantum,baker2020density} to find the coefficients of the end result. 

\begin{proposition}[Practical scaling]
In the case where $X+\frac{L\epsilon}N\Delta$ must be implemented by a decomposition onto $n$ qubits, the scaling of the overall algorithm increases by a logarithmic factor $\log_2 N$--where the input operator is efficiently implementable--to $O(\frac k{\varepsilon^2}\log_2N)$.
\end{proposition}

A linear combination of unitaries \cite{childs2017quantum,low2019hamiltonian} can potentially be used to decompose the operator in to $n=\log_2N$ gates. 
%
This assumes that $X$ is efficiently encodable on the quantum computer; however, the assumption for $L$ to be well controlled in the time evolution would require care to implement. Throughout we have considered $X$ to be representable as a sparse-rank input. The energy eigenvalues have been written so far to be represented to a sufficient approximation, but it is only in Appendix~\ref{qgaimplement} where the full imprecision of the QPE is present.

\subsection{Error}

When describing the sum of terms composing the expectation value of the pseudo-inverse, it has been suggested that
\begin{equation}
\langle\Phi|X^{-1}|\Phi\rangle=\sum_{p=1}^N\frac{|\langle p|\Phi\rangle|^2}{E_p}
\end{equation}
is a suitable starting point, derived from Eq.~\eqref{measY}. Separating the summation, we can define a value $k$ that the series is summed and treat the remainder as a bound. Let
\begin{align}
\langle\Phi|X^{-1}|\Phi\rangle-\sum_{p=1}^k\frac{|\langle p|\Phi\rangle|^2}{E_p}=\sum_{p=k+1}^N\frac{|\langle p|\Phi\rangle|^2}{E_p}
\end{align}
and note that we can take the absolute value of both sides to find
\begin{equation}\label{errormeasure}
\left|\langle\Phi|X^{-1}|\Phi\rangle-\sum_{p=1}^k\frac{|\langle p|\Phi\rangle|^2}{E_p}\right|\leq\sum_{p=k+1}^N\frac{|\langle p|\Phi\rangle|^2}{|E_p|}
\end{equation}
where the inequality appears from the triangle inequality. Eq.~\eqref{errormeasure} constitutes an effective difference between the target value and the sum of the absolute value of the discarded pieces.

\subsection{Discussion}

The poly-logarithmic scaling of this algorithm is dependent on writing a suitable operator controlled on the auxiliary qubits. Note that only one loop is required over $k$ eigenvalues giving the scaling of $O(k)$ in terms of the number of times expectations values must be taken on the quantum computer. This technique is already well studied in many works and the application to gates in the QGPE (which have the same form as time-evolution gates) is already studied, reducing the barrier to implementation for this algorithm. This does mean that well-explored methods to apply unitary gates on the quantum computer can immediately be applied in this algorithm to achieve an efficient implementation.

The existence of this algorithm seems to match the expected optimal scaling by related works on the quantum singular value transformation \cite{chia2022sampling,bakshi2024improved,tang2019quantum}. 

The quantum algorithm is less than classical algorithms that we are aware of by an exponential factor \cite{kulis2006learning} in the case of the practical scaling ($N\rightarrow \log_2N$). Classical algorithms that have pursued similar strategies \cite{bach2013sharp,scetbon2021low} have been reported to have $O(k^2N)$ scaling \cite{kulis2006learning}, meaning that the algorithm here is less by an exponential factor (a ratio of $\frac{\log_2N}{kN}$ for some precision) in the practical case. 

This algorithm can be viewed as a dual algorithm to Ref.~\onlinecite{harrow2009quantum} in the sense that Ref.~\onlinecite{harrow2009quantum} focuses on the input wavefunction rather than the operator \cite{clader2013preconditioned,childs2017quantum,wossnig2018quantum,HHLSurvey,tang2018quantum,chia2020sampling,tang2021quantum,aaronson2015read,costa2022optimal}. The algorithm here is focused on treating the operator itself. We can only speculate on exactly what operators will be best for this algorithm and so cannot comment on exactly how efficient the implementation will be in light of future progress. We do note that writing the operator in a format amenable to the quantum computer may require a different form for the QGPE; we merely followed Ref.~\onlinecite{jordan2005fast}'s form in this paper. 

It has been noted before that similar algorithms \cite{harrow2009quantum} are highly dependent on their state preparation assumptions \cite{tang2019quantum}, and this dual algorithm is likely highly dependent on the ability to implement input time-evolution gate operators in an efficient way. However, for many cases of interest, this problem is a common feature of many algorithms, so improved strategies with implementing time-evolution gates is likely to occur in other areas and may enable more use applications here. We suggest that this feature, being common to many other algorithms, may have a clearer path to improvement over the wavefunction preparation techniques required for other algorithms \cite{tang2019quantum,aaronson2015read}. The gate depth will be entirely dependent on the (controlled) gates that are applied in acting the operators of QGPE onto the quantum machine and their implementation \cite{mottonen2004transformation}. 

We further note that Ref.~\onlinecite{harrow2009quantum} makes use of a sparsity assumption for the input operator. We can in principle do so here, but may have to compute many derivatives regardless of the sparsity assumption. A block form of the input operator would be a natural equivalent here.

\section{Applications to sparse-rank input: kernel-based machine learning}

There are several applications for which the pseudo-inverse of a matrix would be useful. We list a few interesting ones, although this would not constitute an exhaustive list. We start where the algorithm is best ({\it i.e.}, when the problem is machine-learnable and has few eigenstates necessary for a complete description). It is then also pointed out that the problem can be applied on non-sparse-rank matrices for interacting quantum-many problems and beyond.

\begin{definition}[Machine-learnable]
An input matrix is defined to be machine-learnable when the sum of the $k\lll N$ largest eigenvalues has a small difference from the sum over all eigenvalues.
\end{definition} 

We will give a specific case here for the learning of a kernel-based machine learned function \cite{hofmann2008kernel,davis2007information}, but in general the expressions can be more generic.

One can construct a kernel matrix $\mathbf{K}$ between the data points $x_j$ with elements
\begin{equation}
K_{ij}=\kappa(x_i,x_j,\sigma)
\end{equation}
where an example might be, if choosing an exponential function to be $\kappa(x_i,x_j,\sigma)=\exp(-\|x_i-x_j\|^2/\sigma^2)$ where any function may be chosen in general and there can be more hyperparameters than the single $\sigma$ given here.

The machine-learned approximation is given by \cite{liPRB16}
\begin{equation}
f^\mathrm{ML}(x)=\sum_{j=1}^{N_T}\alpha_j\kappa(x,x_j,\sigma)
\end{equation}
for parameters $\alpha_j$ that we will determine by first constructing a cost function to be minimized of the form
\begin{equation}
\mathcal{C}(\boldsymbol{\alpha})=\sum_{j=1}^{N_T}\left(f^\mathrm{ML}(x_j)-f_j\right)^2-\lambda\boldsymbol{\alpha}^T\mathbf{K}\boldsymbol{\alpha}
\end{equation}
where the elements of $\boldsymbol{\alpha}$ are $\alpha_i$ and are solved by
\begin{equation}
\boldsymbol{\alpha}=(\mathbf{K}+\lambda\mathbb{I})^{-1}\boldsymbol{f}
\end{equation}
which requires the pseudo-inverse of a matrix and training points $f_j$ as elements of $\boldsymbol{f}$. The hyper-parameters $\sigma$ and $\lambda$ must be chosen to fit the model properly ({\it i.e.}, avoid over-fitting). In general, one may construct a more generic cost function to minimize.

The appearance of an pseudo-inverse is exactly what we wish to solve with the new algorithm proposed in this paper. The number of relevant principal components in the kernel matrix depends on the application, but in general, we can say that a system is machine-learnable if the number of principal components is small.


\begin{lemma}[Super-linear convergence]

Given an input matrix $X$ such that the singular value $\sigma_i\rightarrow0$ as $i$ increases where $i$ indexes the singular values in order. Given $S$ determined out to an order $k$ sparse-rank factorization, $S^{(k)}$ (where $S^{(N)}=X$), then the approximate singular values $\tilde\sigma_i(S^{(k)})$ tend to
\begin{equation}
\tilde\sigma_i(S^{(k)})\rightarrow\sigma_i
\end{equation}
superlinearly in $k$ ({\it i.e.}, as a high degree polynomial in $k$).

\end{lemma}

The proof is contained in Ref.~\onlinecite{yuan2018superlinear} (Theorem III.5 of that paper) \cite{yuan2018convergence,saad1980rates,li2015convergence}.

This implies that the choice of block Lanczos as a starting point for kernel-based machine learning is highly efficient with respect to the number $k$ eigenvectors that need to be supplied. This algorithm would result in much larger systems being able to be machine-learned.

We also note that the methods here once implemented would allow for the training of Bayesian machine learning algorithms to predict likely outcomes from the machine-learning model \cite{fitzsimons2017bayesian,zhao2019bayesian,mackenzie2014bayesian,BayesianDeepLearning}.

We note that at least $k$ relevant eigenvectors must be present for the algorithm to obtain an answer in general. We find the most useful complexity proof demonstrating this to be native to machine learning. It has been proven in Ref.~\onlinecite{servedio2004equivalences} (with concurring results in at least Ref.~\onlinecite{arunachalam2018optimal} and Ref.~\onlinecite{huang2021information}) that the number of oracle queries for a quantum machine learning algorithm are equal to the number of training points on the classical computer ({\it i.e.}, quantum computers store no more correlations than are provided to it; same for the classical computer). So, at least this much information must be provided in the wavefunction, since Ref.~\onlinecite{servedio2004equivalences}'s results follow from quantum machine learning under the probably approximately correct (PAC) model of learning \cite{valiant1984theory,amsterdam1988valiant}. 

\section{Extensions for other state preparation assumptions}\label{sumQGLDsec}

\begin{figure}
	\begin{algorithm}[H]
		\caption{Summed Quantum gradient of the logarithm-determinant ($\Sigma$-QGLD)}
		\label{sumQGLD}
		\begin{algorithmic}[1]
			\State Prepare a register of an equal superposition of $k$ eigenstates \Comment{Eq.~\eqref{equalsuperposition}}
			\State Generate the factor $1/E_p$ on $\epsilon$ \Comment{Eq.~\eqref{EpPhase}}
			\State Run an inverse time evolution \Comment{Eq.~\eqref{invtimeEvolve}}
			\State Apply $\mathcal{U}_c$ as in the QGLD algorithm but on the superposition of all eigenstates to generate $\delta E(x_{ij})$ for all required elements \Comment{Eq.~\eqref{Uc}}
			\State  \textbf{return} Measure the expectation value of some operator to obtain the sum over all $k$
		\end{algorithmic}
	\end{algorithm}
\end{figure}

So far, we have made conservative assumptions about the preparation of the initial wavefunction. This forced the resulting algorithm to scale as $O(k)$. It is natural to wonder if the algorithm could be reduced in scaling from $O(k)$ to $O(1)$. Doing so would promote the pseudo-inverses in the previous paragraphs to full inverses since all eigenvalues would be represented. We discuss this possibility here if we are given an equal superposition of eigenstates.  

To do this, we would require 
\begin{equation}\label{equalsuperposition}
|\Psi\rangle=\sum_{p=1}^Nc_p|p\rangle
\end{equation}
with $c_p=1/\sqrt{2^n}$ $\forall p$. This is not the same as an equal superposition of all states from a Hadamard transformation.

We note that this is related to the same wavefunction preparation assumptions as other works, namely Ref.~\onlinecite{harrow2009quantum} \cite{tang2018quantum,chia2020sampling,tang2021quantum,aaronson2015read}. If a suitable set of preparation assumptions can be imposed (potentially with a quantum memory), then we will explore that possibility in this section.

If the eigenvalues are known, then this procedure can become simpler. One could in principle construct a Hadamard gate in the basis of the eigenfunctions and prepare the state accordingly.

Our goal is only to illustrate the steps required to generate the gradients in a superposition as a summed-QGLD ($\Sigma$-QGLD) which is summarized in Box~\ref{sumQGLD}.

\subsection{Preparation of an equal superposition of eigenstates}

Considering the wavefunction $|\Psi\rangle$, take a QPE to find
\begin{equation}
\mathcal{U}_\mathrm{QFT}^\dagger\mathcal{U}_\mathrm{QPE}|\Psi\rangle|0\rangle=\frac1{\sqrt{2^n}}\sum_{p=1}^N|p\rangle|E_p\rangle
\end{equation}
where a QFT was applied to bring the phase into an auxiliary register. The number of qubits for the auxiliary register is $m$, a small number to store the precision of the eigenvalue $E_p$.

\subsection{Inverse phase generation}

When performing the the QGLD in superposition, the eigenvalue must be divided by the perturbing value $\epsilon\rightarrow\epsilon/E_p$ as in Eq.~\eqref{equivlogdet}. That is, we need 
\begin{equation}\label{EpPhase}
e^{i\frac N{WL}\left(\hat H+\frac{L\epsilon}N\Delta\right)}\rightarrow e^{i\frac N{WL}\left(\hat H+\frac{L\epsilon}{E_pN}\Delta\right)}
\end{equation}
for a given eigenvalue $E_p$ of the Hamiltonian.  Even if this procedure is more expensive than normal, or requires a detailed implementation \cite{thapliyal2019quantum}, then the rest of the algorithm is small enough to justify the overhead. Note that the eigenbasis of the inverse operator is the same for the input operator, $X|p\rangle=E_p|p\rangle\Leftrightarrow Y|p\rangle=\frac1{E_p}|p\rangle$.

\subsection{Eigenphase oscillation in superposition}

If Eq.~\eqref{equalsuperposition} is provided, then we can evaluate the entire sum over all $k$ in a single oracle query of the QGLD. Upon evaluation of the superposition, there will be no uniform global phase. Instead, terms like $\exp(2\pi i(\theta^{(p)}_0-\theta^{(p')}_0))$ will appear and prevent the proper evaluation of the measurement. Instead, we must cancel the phase on each eigenstate individually in the sum.

To cancel the phases on each state, we can apply an inverse time-evolution operation of the form
\begin{equation}\label{invtimeEvolve}
\mathcal{U}=\exp\left(-\frac {2\pi iN}{WL}X\right)
\end{equation}
which will give 
\begin{equation}\label{equalsuperposition_phase}
|\Psi\rangle=\frac1{\sqrt{2^n}}\sum_{p=1}^Ne^{-{\frac {2\pi iN}{WL}\theta^{(p)}_0}}|p\rangle
\end{equation}
and cancel the (``global")\footnote{But now the previously defined ``global phase" is inherent to each eigenstate, meaning it is not so global in this case. It is different for each eigenstate in the superposition.} phase on each state $p$ when applying Eq.~\eqref{timeEvolve}.

\subsection{Complexity}

The overall scaling of the method reduces to $O\left(\frac1{\varepsilon^2}\right)$ (ideal) or $O\left(\frac1{\varepsilon^2}\log_2N\right)$ (practical),\footnote{We note throughout that we quote the standard implementation of the QFT \cite{kitaev1995quantum,nielsen2010quantum}. However, there is a more efficient version in Ref.~\onlinecite{hales2000improved} that is only $O(\frac1\varepsilon\log_2\frac1\varepsilon)$ that could in principle be applied to reduce the scaling.} a reduction by a factor of $k$ from the QGLD. This scaling likely hides the complexity of preparing the wavefunction \cite{aaronson2015read} and the inclusion of the inverse eigenvalue as discussed in the previous paragraph. This may require a quantum RAM (QRAM), although it has been noted the difficulties in using this device and constructing it.

\subsection{Application of $\Sigma$-QGLD}

For completeness, we discuss here the use of a non-equal superposition of eigenstates. Consider a generic wavefunction, perhaps only restricted by its ease of preparation on the quantum computer,
\begin{equation}\label{unequalsuperposition}
|r\rangle=\sum_{i=1}^Nc_p|p\rangle
\end{equation}
where $c_p\in\mathbb{C}$. This is not the equal superposition of Eq.~\eqref{equalsuperposition}, but we can propose a scheme to sample and recover the coefficients so over many samples, the average of the coefficients becomes uniform.

Note that
\begin{equation}
c_p(r)=\langle r|p\rangle
\end{equation}
which is the coefficient in Eq.~\eqref{unequalsuperposition} for a given input vector $r$. Choosing a number of random initial states $|r\rangle$ would sometimes create coefficients $c_p(r)$ that sometimes are less than and sometimes greater than the equal value desired for Eq.~\eqref{equalsuperposition}.

The average over many choices of $r$ may average to the value needed in Eq.~\eqref{equalsuperposition}, particularly if the probability amplitudes of the spectral decomposition of $|r\rangle$ are distributed such that they are distributed around $1/\sqrt{2^n}$. Thus, averaging the results of many such choices of $r$ would approximate the solution of Eq.~\eqref{equalsuperposition}. The algorithm is simply developed from the previous introductions. Box~\ref{QGLD-sampling} suggests to prepare several states $r$ and average the results.

The primary hope of this construction is that the sampling does not scale with the size of the problem. Each eigenvalue receives a projection from the initial state $|r\rangle$, meaning that the individual terms are uncorrelated from sample to sample merely by construction. 

The scaling of this algorithm would be the same as the $\Sigma$-QGLD times the number of samples required. We mainly explore the sampling as though it were performed individually on each $|r\rangle$ state; however, this does not mean that the expectation value could not be taken at once, potentially with a fast tomography technique \cite{aaronson2019shadow},  but we would prefer to examine an implementation in full in a future work.

In tests of this algorithm on four models (spin-half, Hubbard, $t-J$, and the transverse field Ising model) of a small enough size to solve with exact diagonalization, we noticed a slight increase in the error with the system size. The hope was that the average over many samples would tend towards an equal superposition, but this was not perfectly convergent in our small implementation.  We include these details in case another attempt is tried as a sampling technique may prove useful for this method.

This method is perhaps best applied when the eigenstates are known, such as for the uniform gas. In this case, we take a collection of fermions in the unit cell of an infinite system \cite{fetter2012quantum}. The external potential is set to cancel a divergence in the size-extensive exchange term (although, this can be taken to be zero for all practical purposes), and the eigenstates are considered to be plane waves $\exp(i\mathbf{k}\cdot\mathbf{r})$. This has been explored on the quantum computer before \cite{ivanov2024quantum,bauer2020quantum,babbush2018low,mcclean2020openfermion}. In this case, we can prepare the set of eigenstates on the quantum computer through a unitary operation.

\begin{figure}
	\begin{algorithm}[H]
		\caption{QGLD sampling}
		\label{QGLD-sampling}
		\begin{algorithmic}[1]
			\State Prepare a number of orthogonal random states in superposition
			\State  \textbf{return} Average many samples of the $\Sigma$-QGLD algorithm
		\end{algorithmic}
	\end{algorithm}
\end{figure}

In this case, we can then evaluate functional derivatives of the quantum field theory as in Eq.~\eqref{correlationfct} by using $\Sigma$-QGLD as an oracle query to the QGA. This will evaluate the problem over a set number of qubits we use on the quantum computer. This implies a plane-wave cutoff that scales as the number of qubits \cite{helgaker2014molecular}. In principle, one could evaluate for exchange and correlation function.\footnote{For example, see Eq.~\eqref{interactingGreensfct}.} Some existing results can be checked against quantum Monte Carlo studies \cite{umrigar1994accurate}.

Annealing the eigenstates of the uniform gas to a slightly-perturbed Hamiltonian would allow for study on any Fermi-liquid \cite{fetter2012quantum}.

\section{Conclusion}

Logarithm-determinants were shown to be efficiently computable on the quantum computer from the knowledge of the derivative of the eigenvalue and the eigenvalue itself if extremal eigenstates are known. The logarithm-determinant's derivatives apply widely in both field theoretic formulations of physics as well as machine learning.

A multi-variable series expansion was derived from a quantum gradient algorithm. It was shown that applying a time-evolution operator in $O(1)$ query time can generate the contribution from a single eigenvalue to the expectation value of an inverse operator. A sum over $k$ such eigenvalues is required in order to obtain the full result. The efficient implementation of the operator used here--based on the quantum gradient algorithm--will be considered in future studies.

For practical implementation, the input matrix into the argument of the time-evolution operator can be written as a linear combination of unitaries and decomposed to give a poly-logarithmic scaling algorithm in terms of the size of the input matrix, $\log_2N$. We also provided the method extended to a superposition of eigenstates provided that the initial equal superposition of states can be created.

In those cases where the number of eigenvectors needed to well-describe the input operator is small, the sparse-rank factorization of the matrix is known to converge super-linearly in the rank of the matrix for sufficiently decaying singular values, implying that the method is very efficient in those cases. One potential use case is therefore quantum machine-learning.

The goal was to show that the logarithm-determinant can be solved on the quantum computer. There is no claim about the ultimate performance with respect to all possible unknown classical algorithms \cite{tang2018quantum,chia2020sampling,tang2021quantum,aaronson2015read}.

\section{Acknowledgements}

We thank Olivia Di Matteo, Nishanth Baskaran, Sarah Huber, Michael O.~Flynn, James Brown, Arif Babul, and Emmanuel Fromager for useful comments. T.E.B. is forever indebted to David Poulin for the original inspiration and time that led to this work. 

We are indebted to the anonymous referee, especially for pointing out Eq.~\eqref{errormeasure} out to us and the surrounding section.

We also thank the members of the Quantum Computing for Sustainability Working Group for stimulating discussions.

Some figures were prepared with the Quantikz package \cite{kay2018tutorial}.

\paragraph*{Note added after submission.--} The final result in this paper of the preparation of an equal superposition of eigenstates was effectively accomplished with an efficient operator implementation in Ref.~\onlinecite{ethsum}. Interestingly, the other concept of a directional derivative--where a matrix is the result of the function instead of a scalar as here--was the key to find a more efficient implementation.

\section{Funding}

J.G.~acknowledges the NSERC CREATE in Quantum Computing Program (Grant Number 543245). J.G.~is supported in part by funding from the Digital Research Alliance of Canada.

T.E.B. is grateful to the US-UK Fulbright Commission for financial support and being hosted by the University of York.  This research was undertaken in part thanks to funding from the Bureau of Education and Cultural Affairs from the United States Department of State.

This research was undertaken, in part, thanks to funding from the Canada Research Chairs Program (CRC-2021-00257). This work has been supported in part by the Natural Sciences and Engineering Research Council of Canada (NSERC) under grants RGPIN-2023-05510, DGECR-2023-00026, and ALLRP 590857-23. 

This work is made possible in part with support from the University of Victoria's start-up grant from the Faculty of Science. 

\section{Data Availability Statement}

All implementations in small circuits were accomplished with the DMRjulia library \cite{bakerCJP21,*baker2019m,dmrjulia0,dmrjulia,dmrjulia1}. The methods can also be implemented with exact diagonalization on dense matrices.

\section{Conflict of interest declaration}

The authors declare that they have no conflict of interest.

\section{Author Contributions}

J.G.~performed all implementation steps. T.E.B.~developed the project. The article was co-written and authors are listed alphabetically.

\onecolumngrid

\bibliography{QBLR,RQBLM,TEB_papers,TEB_books,lanczos_refs}

\begin{appendix}

\section{Logarithm-determinant in physics}\label{logdet_phys}

The logarithm-determinant appears in many places in the study of physics. Considering the generic expression for a partition function or generating function of the form
\begin{equation}
Z=\int\mathcal{D}\varphi \;e^{i\mathcal{S}[\varphi]}
\end{equation}
where the functional integral is over all possible fields, $\varphi$ and $\mathcal{S}[\varphi]=\int\mathcal{L}(\varphi,\dot\varphi,t)dt$ for a given Lagrangian, $\mathcal{L}$.  The Lagrangian may describe a classical or quantum system at zero or finite temperature $T$. The quantity $Z$ can be evaluated in a basis and then recast as a matrix, albeit one that is practically exponentially large in order to obtain accurate results. 

Knowledge of $Z$ can be used to obtain thermodynamic quantities which are related by the derivative of the logarithm to obtain \cite{reif2009fundamentals}
\begin{equation}
E=-\frac{\partial}{\partial\beta}\ln Z
\end{equation}
the energy of a system for one example, for an inverse temperature $\beta=1/T$ (allowing units such that the Boltzmann constant is 1). In general, taking the derivative of the partition function can give the conjugate variables in thermodynamics \cite{reif2009fundamentals}.

In the context of a field theory, one can generate correlation functions by taking successive derivatives of source fields that are added to the Lagrangian ({\it i.e.}, $J\varphi$). Taking functional derivatives of the source fields $J$ such as \cite{schwartz2014quantum,peskin1995quantum}
\begin{equation}\label{correlationfct}
\left.\frac{\partial}{\partial J(x)}\left(\frac{\partial}{\partial J(y)}\ln Z\right)\right|_{J=0}=\langle\hat c^\dagger_{i\sigma}(x)\hat c_{j\sigma}(y)\rangle
\end{equation}
would be the Green's function when $J=0$. 

Fully interacting Green's functions can also be obtained from \cite{fetter2012quantum,abrikosov2012methods}
\begin{equation}\label{interactingGreensfct}
\mathcal{G}_{ij}(\omega)=\langle\Psi|\hat c^\dagger_{i\sigma}(\omega-\hat H\pm i\eta)^{-1}\hat c_{j\sigma}|\Psi\rangle
\end{equation}
where $\eta\rightarrow0$. The appearance of the inverse of an operator aligns with the scope of the main results.

We note one small equivalence in the literature that is equivalent to the logarithm-determinant.

\begin{theorem}[Trace-logarithm form]
The logarithm-determinant is equvialent to the trace of the logarithm of an operator,
\begin{equation}
\ln\det X=\mathrm{Tr}\ln X
\end{equation}
\end{theorem}

\begin{proof}
This is straightforwardly derived by substituting an eigenvalue decomposition--as in Eq.~\eqref{sparserankX}--and rearranging.
\end{proof}

\section{Gradients of the logarithm-determinant}\label{classicalchecks}

We can compute a variety of examples with the technology that we have implemented through a simulator for use in checking future results on a quantum computer and to illustrate the concepts.

For all examples, we fix
\begin{equation}
\Delta\propto\sigma^x
\end{equation}
with elements $\Delta_{ij}$, and we restrict to a single qubit. We focus on the derivative with respect to the $x_{12}$ element in each case.

Derivatives in this section are taken with respect to finite difference approximations of the perturbed matrix (non-quantum implementation). The full quantum implementation, which matches the exact result roughly to $O(L)$, is detailed in Appendix~\ref{qgaimplement}. The scaling as $O(L)$ only means that the leading correction to the answer decreases as $L$.

\begin{figure}
\includegraphics[width=0.5\columnwidth]{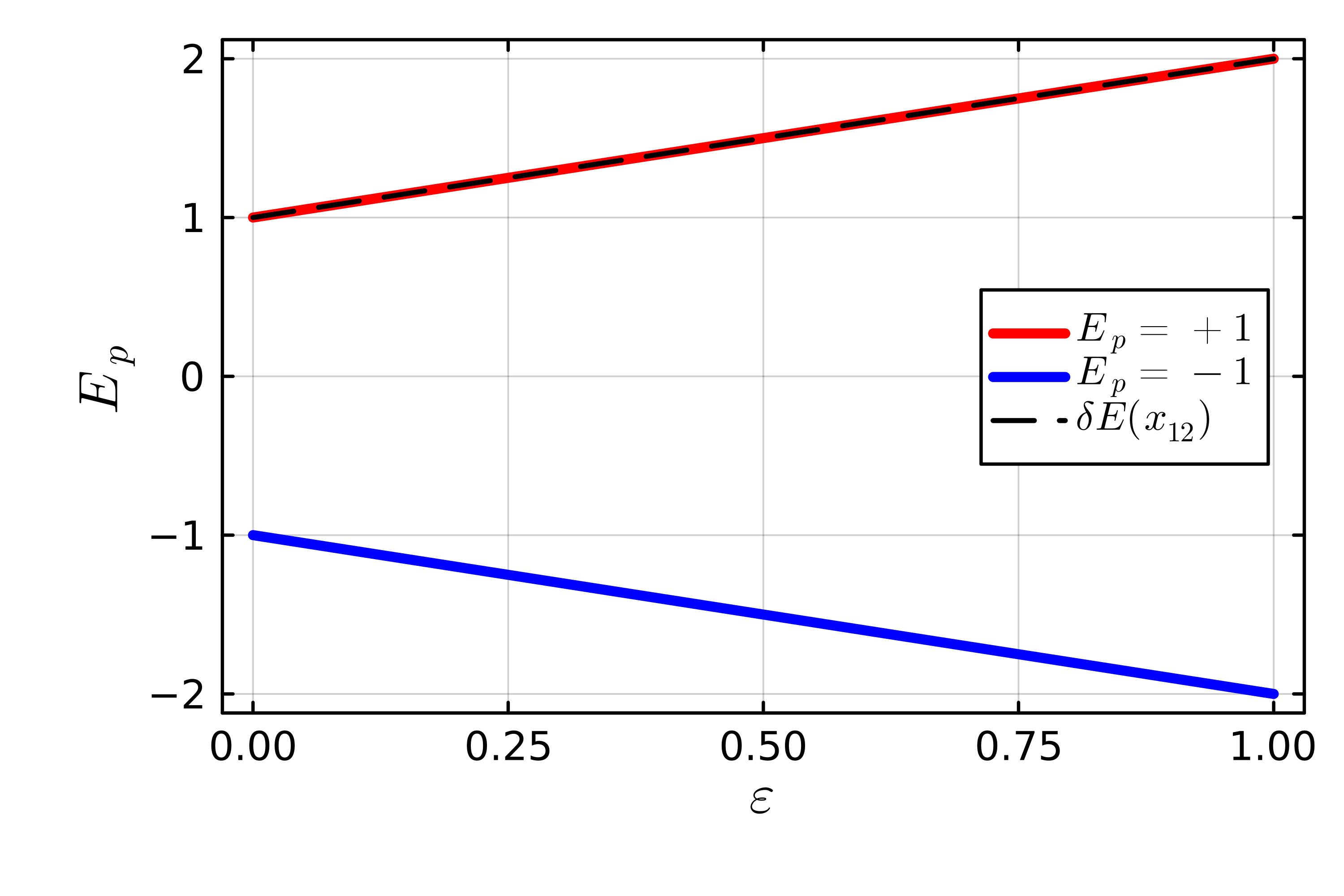}
\caption{
An input operator of $\sigma^x$ given eigenvalues $\pm1$ can be shifted by a matrix $\Delta=\epsilon\sigma^x$. The eigenvalues of the shifted matrix are shown here. The derivatives make for a good way to check the algorithm on a small example. In this case, the eigenvalue lies on top of the derivative of that eigenvalue with respect to $x_{12}$.
\label{ox_example}
}
\end{figure}

\subsection{$X=\sigma^x$}

As a basic example, we compute the derivative of the Pauli-matrix
\begin{equation}
\sigma^x=\left(\begin{array}{cc}
0 & 1\\
1 & 0
\end{array}\right)
\end{equation}
with the oracle query developed in Sec.~\ref{oracle_section}. This means that $X=\sigma^x$. In this example, we have that $\epsilon\Delta_{ij}=\epsilon\sigma^x$ because we want to uncover the derivatives of the eigenvalues in this case. The results over a range of $\epsilon$ are shown in Fig.~\ref{ox_example}.

The inverse of $\sigma^x$ is of course itself because it is unitary. Thus, the answer must be equal to 1. As can be seen in Fig.~\ref{ox_example}, this is the correct answer. Since the eigenvalue of the largest eigenvalue is 1 as well, the logarithm-determinant reduces to just the eigenvalue, hence why this example was chosen.

The quantum algorithm was implemented in a simulator in the DMRjulia library \cite{bakerCJP21,*baker2019m,dmrjulia0,dmrjulia,dmrjulia1}, and the result when using a single qubit register ($m=1$) matches the slope in Fig.~\ref{ox_example}.

\subsection{$X=\sigma^z$}

When giving the input of $\sigma^z$, we again apply $\Delta=\epsilon\sigma^x$. Since the inverse of $\sigma^z$ is also itself, we should expect that this element will be zero. As the result of Fig.~\ref{oz_example} shows, this is true in this case.

\begin{figure}
\includegraphics[width=0.5\columnwidth]{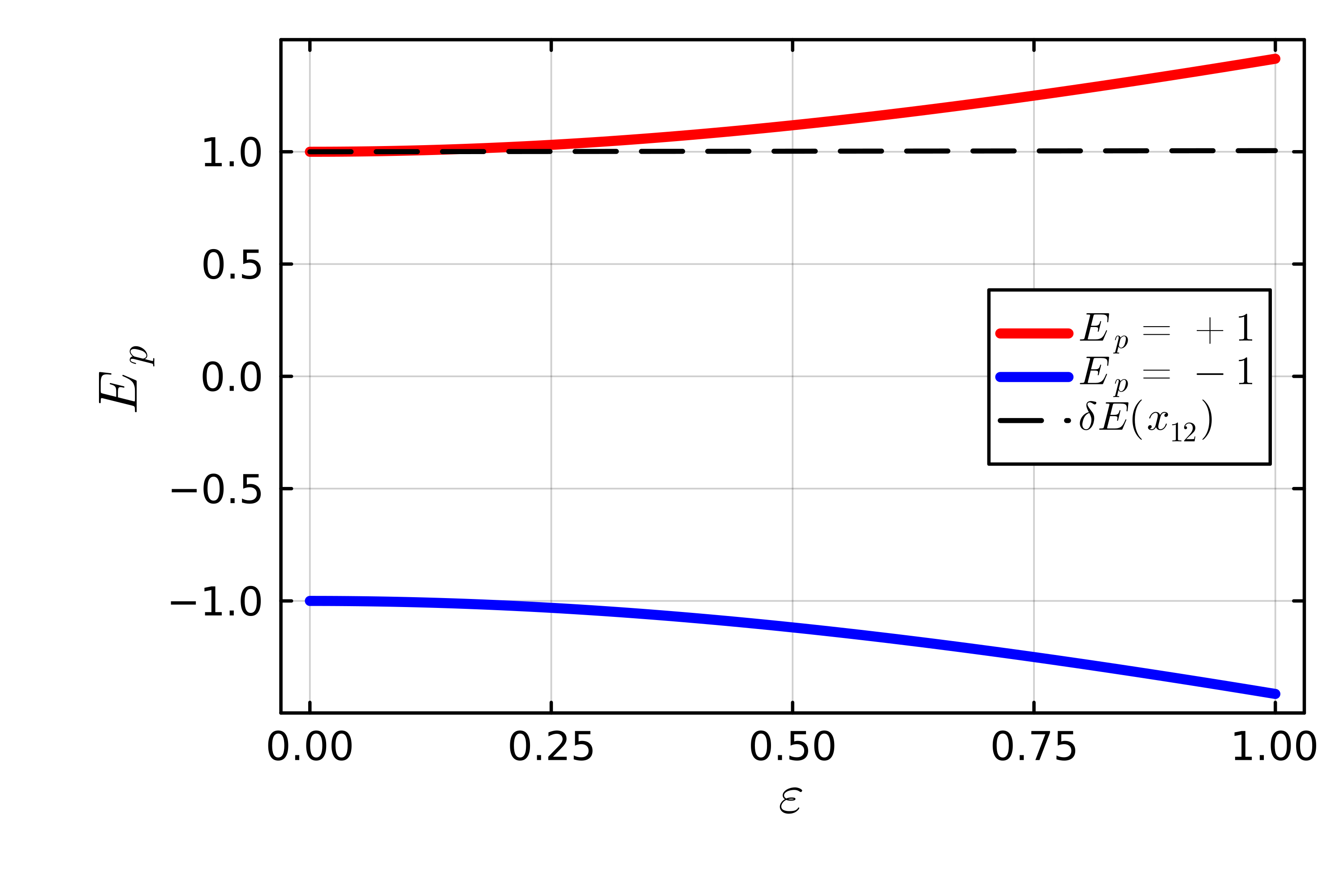}
\caption{
An input operator of $\sigma^z$ given eigenvalues $\pm1$ can be shifted by a matrix $\Delta=\epsilon\sigma^x$. The eigenvalues of the shifted matrix are shown here. The derivatives make for a good way to check the algorithm on a small example.
\label{oz_example}
}
\end{figure}

The quantum simulator also reproduces this result.

\subsection{$X=H$}

\begin{figure}
\includegraphics[width=0.5\columnwidth]{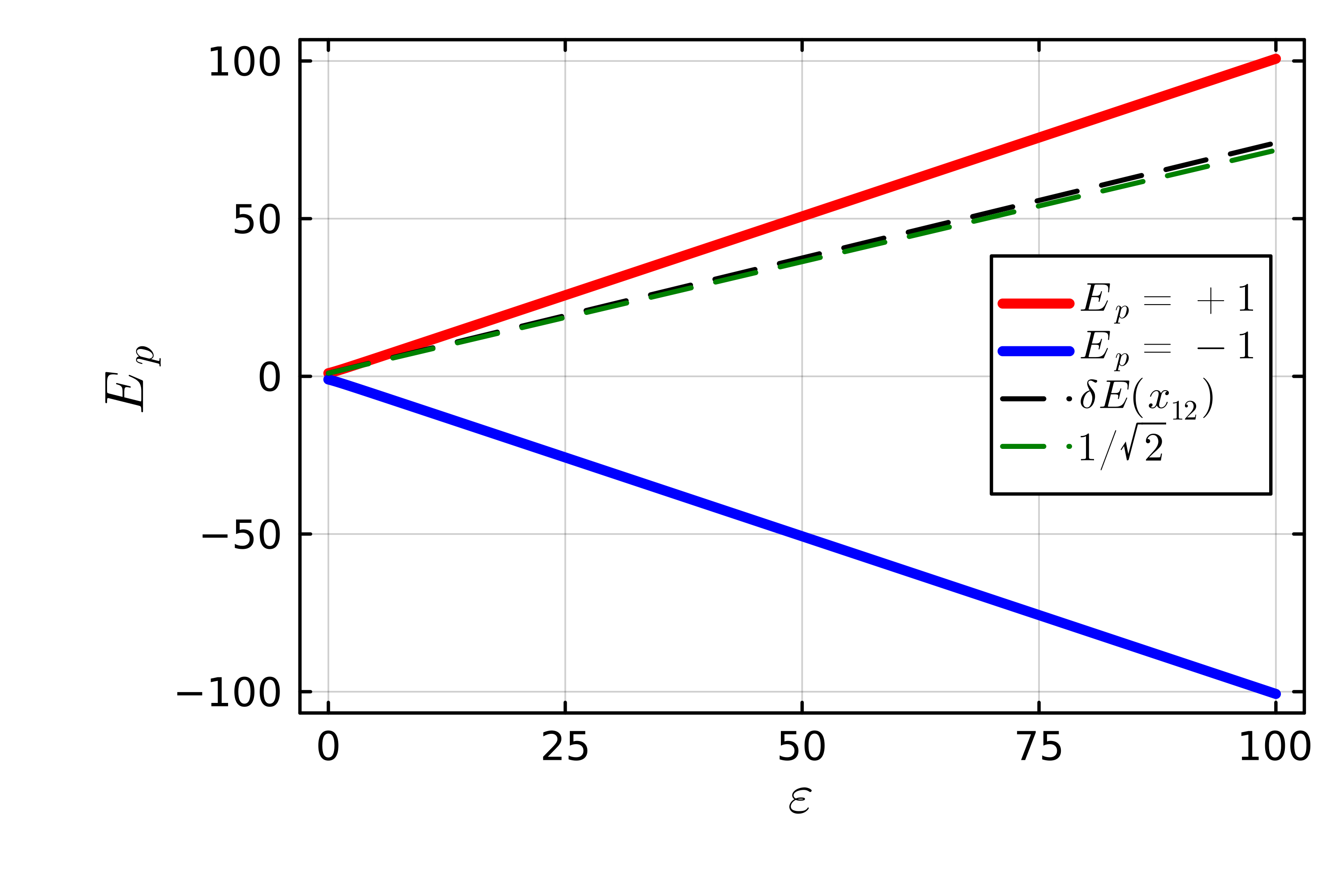}
\caption{
An input operator of $H$ given eigenvalues $\pm1$ can be shifted by a matrix $\Delta=\epsilon\sigma^x$. The Hadamard operator is its own inverse, thus we expect a value of $1/\sqrt2$. The grid discretization $\delta\varepsilon=10^{-1}$ causes some deviation in the slope due to a finite difference method used to check the results but it is not present in the quantum algorithm \cite{jordan2005fast}. On this scale the values are nearly identical, but it is an important difference between the values when performing practical computation.
\label{hadamard_example}
}
\end{figure}

We also can see what happens when we use a Hadamard gate and see what the slopes are. In this case, the Hadamard is also its own inverse; however the eigenvalues are still $\pm1$. We therefore expect a slope of $1/\sqrt2$, matching the inverse element. If we choose a grid spacing of $\delta \epsilon$, then we can see that the deviation between the finite difference ($\delta\epsilon^{(\mathrm{classical})}=0.73...$) between the first two points\footnote{One can in principle extrapolate in the slopes instead of just taking the first two points, but we wanted to illustrate the concept from Ref.~\onlinecite{jordan2005fast}. Automatic differentiation could be used here too.} and quantum gradient results has a deviation on the order of ($\delta\epsilon^{(\mathrm{quantum})}=1/\sqrt2$).

\section{Negative determinants and eigenvalues}\label{neg_eigvals}

There is a question of the applicability of Eq.~\eqref{logdetfirst} when the determinant of a matrix is negative. The natural logarithm is defined only in the complex plane by $\ln z=\ln|z|+i\theta$ where $\theta$ is chosen as a branch cut arbitrarily. Often, this is chosen to be along the negative real axis, but this can be chosen elsewhere (for example, $\theta=3\pi/2$). In our application here, we were interested in Hermitian operators.\footnote{We can extend the concepts on the quantum computer to non-Hermitian operators for $X$ by first creating an operator that is Hermitian of the form $A=(X+X^\dagger)/2$ and $B=(X-X^\dagger)/(2i)$ and finding $X=A+iB$.} It is possible that the angle could change if the eigenvalue were to change signs, but there are constructions to guard against this.

Another way to construct the problem is to always guarantee that the determinant is positive. Representing the logarithm in the complex plane, we do acquire an absolute value sign in the leading term $\ln|z|$ which would potentially given an ambiguity. Another way to set up the problem would be to take the gradient of the logarithm-determinant of $W=-X$ so that all eigenvalues change sign in the matrix. This changes the sign in the determinant overall. To still obtain the same answer, we must take either $-\partial_{w_{ij}}$ where $w_{ij}=-x_{ij}$ or still take the partial derivative with respect to the elements of $X$. In examining Eq.~\eqref{logdetequiv}, the exchange of $E_p\rightarrow-E_p$ does not change the overall result, verifying this approach.

To verify the negative eigenvalues work, even if the negative determinant appears, the evaluation from Eq.~\eqref{logdetequiv} still holds. As an example, considering the case of $X=\sigma^z$, we have
\begin{equation}
\sigma^z=\left(\begin{array}{cc}
1 & 0\\
0 & -1
\end{array}\right)
\end{equation}
with eigenvectors $|0\rangle=(1,0)^T$ and $|1\rangle=(0,1)^T$. The evaluation of $\langle p|e^{2\pi i\frac N{WL}(Z+\frac{L\epsilon}N\Delta)}|p\rangle$ will still produce the correct result despite having negative eigenvalues since Eq.~\eqref{logdetfirst} is represented by Eq.~\eqref{logdetequiv}. The overall sum is zero since $(\sigma^z)^{-1}=\sigma^z$ leading to a full sum being 0 over all $p$. However, we only find positive derivatives of the derivatives of the eigenvalues. We must have $E_p=\pm1$ in the alternative form of Eq.~\eqref{logdetequiv}. Thus, a negative determinant can be handled by the form that we use from Eq.~\eqref{logdetequiv}.


\section{Random quantum block Lanczos recursion}\label{R-QBLR}

We preferred to use a randomized quantum block Lanczos (RQBL) algorithm \cite{yuan2018superlinear,musco2015randomized,wang2015improved,eberly1997randomized,rokhlin2010randomized,halko2011algorithm,gu2015subspace,tropp2022randomized,baglama2006restarted} as a starting point for the QGLD for sparse-rank inputs. Lanczos methods were proposed for the quantum computer \cite{baker2021lanczos} because of their rapid convergence \cite{saad1980rates}, ability to obtain the ground state in consequently fewer operations, and their ability to compute the continued fraction representation of the fully interacting Green's function \cite{baker2021lanczos}. It was also shown that the same techniques using qubitization \cite{low2019hamiltonian,berry2015simulating,kothari2014efficient,childs2017quantum,bakerPRA24} can be used to avoid Trotter decomposition. A block Lanczos algorithm can be implemented and this is advantageous because the cost is equal to square the number of states in each block and the algorithm tends to be more stable \cite{bakerPRA24}. 

Lanczos methods have recently caught the imagination of the quantum computing development community. After the original proposal with block encodings in Ref.~\onlinecite{baker2021lanczos}, there was a nearly immediate demonstration on a real quantum computer of a continued fraction in Ref.~\onlinecite{jamet2021krylov}. The extension to block Lanczos techniques followed shortly after that \cite{bakerPRA24}, noting that the method should be more stable than the scalar Lanczos case based on error estimates. An in-depth summary of the method covered existing strategies in Ref.~\onlinecite{kirby2023exact}--reviewing in detail that the implementation on the quantum computer is efficient--with some discussion of error in Ref.~\onlinecite{kirby2024analysis}. The problem has recently been applied with near-term quantum technologies to large spin chains \cite{yoshioka2024diagonalization} and lattice descriptions of quantum field theories \cite{anderson2024solving}. Some other implementations of Lanczos methods use Trotter decompositions of a time evolution operator \cite{motta2020determining,stair2020multireference,klymko2022real,seki2021quantum}. Some other methods have investigated Davidson methods for excited states \cite{tkachenko2024quantum}.

Even though quantum Lanczos is roughly the same as the classical version, the main justification for its implementation is that the quantum computer bypasses the exponential increase in the storage requirements for the wavefunction. The discussion here will be for the fully error-corrected implementation of quantum computers, just as previous work has been framed \cite{baker2021lanczos}. The limitation of precision on the quantum computer is the uncertainty limit which allows for a more stable Lanczos algorithm and avoids the accumulation of round-off error that appears in classical computations. We note this implies that the wavefunction preparation here is not exponentially costly, a bottleneck of some other algorithms \cite{tang2021quantum,aaronson2015read}. This is ultimately predicated on having only a few relevant eigenvectors necessary for the problem at hand for the QGLD.

We note that other methods of determining excited states could be used as the generator for the eigenvectors.

We consider a set of eigenvectors that is generated by a block-Lanczos algorithm, although any appropriate subspace could be chosen and this algorithm will immediately generaize to that case as well. Previous formulations of this algorithm \cite{bakerPRA24} have focused on starting from a state where the energy with respect to some other Hamiltonian is known. The key difference here is that this algorithm, quantum randomized block Lanczos, will start from a random state prepared on the quantum computer.

Block Lanczos methods have been known to be more stable than scalar Lanczos methods because of their ability to more stably handle degeneracies in the eigenvalue spectrum \cite{bakerPRA24}, avoiding spurious degeneracies that can result from a traditional Lanczos algorithm.

The starting point for this algorithm will be a random initialization of a non-interacting (and non-entangled) wavefunction on the quantum computer. This will be the starting point for all future iterations of the this algorithm. The number of qubits required for this is $n$ to store the wavefunction on the quantum computer.\footnote{Note that the fast counting trick for the Lanczos coefficients used in Ref.~\onlinecite{baker2021lanczos} (and generalizing to block Lanczos) is only possible if the input wavefunction's energy is known as an eigenvalue of the operator in question.}

The traditional formulation of this algorithm on the classical computer is to perform then $X\Psi_0$ ({\it i.e.}, multiplying the matrix onto the wavefunction) as the starting point. However, the discussion as follows does not require anything other than a suitable starting point for the Lanczos algorithm, the block Lanczos algorithm continues exactly as in Ref.~\onlinecite{bakerPRA24}. The multiplication here can be regarded as a way to start relatively closer to the end solution, but it is not necessary.

In all, this represents the initial $b$ states for a block-size $b$ of the Lanczos algorithm. The initial states should be made to be orthogonal to each other. One easy way to create the states on the classical computer is to initialize an $N\times b$ sized matrix $\Psi_0$ randomly and then take a singular value decomposition of the matrix $\mathrm{svd}(\Psi_0)=UDV^\dagger$ \cite{bakerCJP21,*baker2019m} and then reassign $\Psi_0=UV^\dagger$. This will ensure orthogonality. In principle, the initial wavefunction can then be prepared on the quantum computer. 

From the starting wavefunction, we then continue to perform a series of block Lanczos steps as outlined in Ref.~\onlinecite{bakerPRA24}. A block Lanczos algorithm can be formulated as a three term recursion relation
\begin{equation}\label{blocklanczos}
\boldsymbol{\Psi}_{p+1}\mathbf{B}_{p+1}=X\boldsymbol{\Psi}_p-\boldsymbol{\Psi}_p\mathbf{A}_p-\boldsymbol{\Psi}_{p-1}\mathbf{B}_p^\dagger
\end{equation}
where matrices
\begin{equation}\label{An}
\mathbf{A}_p=\boldsymbol{\Psi}_p^\dagger X\boldsymbol{\Psi}_p
\end{equation}
and
\begin{equation}\label{Bn}
\mathbf{B}_p=\boldsymbol{\Psi}_{p-1}^\dagger X\boldsymbol{\Psi}_p
\end{equation}
both of which provide the expectation values that are to be found from the quantum computer to determine the Lanczos coefficient blocks. 

For ease of implementation to find the matrix $\mathbf{B}_{p+1}$, one way to find it is to measure the $b\times b$ elements in the block of $\left(\boldsymbol{\Psi}_{p+1}\mathbf{B}_{p+1}\right)^\dagger\left(\boldsymbol{\Psi}_{p+1}\mathbf{B}_{p+1}\right)=\mathbf{B}_{p+1}^2$ (implicitly making a computation with a real operator more plausible here) and to take the square root of the matrix. All coefficients are stored on the classical computer, so more precision on the classical computer can be stored in order to avoid a less in precision when taking the square root operation.

The matrices $\mathbf{A}_n$ and $\mathbf{B}_n$ are $b\times b$ sub-blocks that create a block-tridiagonal matrix of the form
\begin{equation}\label{rank_k_lanczos}
S=\left(\begin{array}{ccccccc}
\mathbf{A}_0 & \mathbf{B}_1^\dagger & \mathbf{0} & \mathbf{0}\\
\mathbf{B}_1 & \mathbf{A}_1 & \mathbf{B}_2^\dagger & \mathbf{0}\\
\mathbf{0} & \mathbf{B}_2 & \mathbf{A}_2 & \ddots\\
\mathbf{0} & \mathbf{0} & \ddots & \ddots
\end{array}\right)
\end{equation}
where $S$ is the $k$-rank sparse factorization of $X$. 

The operator $X$ does not need to have a special form in order to be applied onto the wavefunctions, it merely must be represented as a linear combination of unitaries. In the seminal work of Ref.~\onlinecite{low2019hamiltonian}, the technique of qubitization defined a pathway to apply not only unitary operations onto the quantum computer, but to convert a broad class of operators to a unitary that can then be applied on the quantum computer. 

\begin{definition}[Sparse-rank]
A matrix $X$ is said to have a sparse-rank ($k$-sparse) if the approximated eigenvalues are similar to the true eigenvalues, $\tilde E_p\approx E_p$, for $k$ such eigenvalues.
\end{definition}

The rate at which the eigenvalues converge to the true eigenvectors is dependent on the algorithm used to find $S$. For the Lanczos algorithm, the rate for extremal eigenstates is very fast if the extremal eigenvalues are not spaced too closely together \cite{saad1980rates}.

The oracle query required for the function $f$ for the case of the logarithm-determinant is given here as a quantum phase estimation controlled on an auxiliary register representing the variations of the input function. 

From the knowledge of the Lanczos coefficients, one can then construct the operator $\Gamma^{(p)}$ to construct each excitation \cite{baker2021lanczos,bakerPRA24}.

Lanczos methods are known to be rapidly convergent to the extremal eigenstates \cite{saad1980rates}, which makes them an excellent candidate for state preparation on the quantum computer. 

The main utility of the Lanczos method is that the wavefunction preparation is not exponentially costly for the defined set of $k$ sparse-rank input wavefunctions that are considered here. This implies that wavefunction preparation is super-linearly convergent \cite{yuan2018superlinear}. Most importantly, this is not exponential, so the cost of wavefunction preparation is not exorbitantly costly for the use case envisioned here. 

Since the quantum computer stores a great deal of precision out to the uncertainty limit, this implies that the Lanczos coefficients can be determined to a great accuracy than on the classical computer and to avoid the escalating effects of finite precision kept on the classical computer. The only limitations are the sampling time spent to obtain Lanczos coefficients and the accuracy of gates applied on the quantum computer (although this goes linearly with the accuracy of the coefficients \cite{bakerPRA24}).

This will be sufficient to determine the most relevant eigenvalues of the input matrix which is useful on its own in a variety of applications. 

Note that we have stored the coefficients of the block Lanczos on the classical computer, so we are able to construct the sparse-rank $S$ classically and that this smaller sized matrix contains the extremal eigenstates after a small number of Lanczos iterations. This gives us access to eigenenergies as well from the diagonalization of $S$.

\section{Quantum gradient algorithm} \label{QGAapx}

In this part, we cover the quantum gradient algorithm (QGA) for an arbitrary smooth function $f: \mathbb{R} \rightarrow \mathbb{R}$ \cite{jordan2005fast}. We describe here in detail for a function $f$ taking only scalar input. This is a simplification of Ref. \onlinecite{jordan2005fast} where $\mathbf{x}$ is a $d$-dimensional vector for a scalar function of interest ($f: \mathbb{R}^{d} \rightarrow \mathbb{R}$). For the $d$-dimensional case, the algorithm computes the derivative with respect to each direction on a separate register.

 An oracle of the function $f$ is first queried over a small linear region of the input domain, the output written onto an ancillary register prepared in a plane-wave state. By phase-kickback, the gradient can be decoded via inverse QFT of the input register. Since the register is first placed into a superposition of evenly spaced inputs, this results in an algorithm that has $O(1)$ query time. The precision of the gradient approximation is dependent on the range of the function over the selected domain, the maximum expected magnitude for the gradient, and the number of qubits in the input register. The simplified circuit for the QGA is presented in Fig.~\ref{QGA}.

\begin{figure}[]
\begin{quantikz}[column sep=14, wire types = {b, b}]
\lstick{$\ket{0}^{\otimes m}$}&\qwbundle[Strike Height= 0.2]{m}&\gate{H^{\otimes m}} &\phase{\epsilon}&\gate{QFT^\dagger}&\meter{} \\
\lstick{$\ket{0}^{\otimes n-1}\otimes|1\rangle$}&\qwbundle[Strike Height= 0.2]{n}& \gate{QFT} &\gate{f(\epsilon)} \wire[u][1]{q} &&
\end{quantikz}
\caption{Circuit for evaluating the gradient of a function $f$.
\label{QGA}
}
\end{figure}

The QGA for the function $f(x)$ with $x\in \mathbb{R}$ begins with just two registers made up of $m$ and $n$ qubits respectively. 
\begin{equation}
|0\rangle^{\otimes m}\left(|0\rangle^{\otimes n-1}\otimes|1\rangle\right)
\end{equation}
A Hadamard transform is applied onto the first register by applying a Hadamard gate to each qubit.
\begin{equation}\label{deviations}
H=\frac1{\sqrt 2}\left(\begin{array}{cc}
1 & 1\\
1 & -1
\end{array}\right)
\end{equation}
The state of the $m$-qubit register in the circuit is written as: 
\begin{align}
\left(H|0\rangle\right)^{\otimes m}
&=\left(\frac1{\sqrt{2^m}}\sum_{\epsilon = 0}^{2^m - 1}|\epsilon\rangle\right)
\end{align}
where $\epsilon$ is a bit-string state of length $m$, and the index of a perturbation value in the domain $ [-L/2, +L/2]$ for which the function is evaluated. As an illustration, the all-zero string for $\epsilon$ represents the perturbation $x = -L/2$ and the all-one string is $x= +L/2$. The transformation between the circuit input state $\epsilon$ and the functional input $x$ is  

\begin{equation}
x = \frac LM\left(\epsilon - \frac M2\right)
\end{equation}
with $M = 2^m$. 

A plane wave state is set-up on the $n$-qubit register which appears as: 

\begin{align}
&\left(\frac{1}{\sqrt{M}} \sum_{\epsilon} |\epsilon\rangle \right) \left(\frac{1}{\sqrt{N}} \sum_a e^{i2\pi a/N} |a\rangle\right) 
\end{align} 
since the qubits were initialized to the value 1 in binary.

The blackbox oracle for QGA first evaluates the function $f$ for the equal superposition state in the $m$-qubit register, then records it onto the second register via modular addition. For this operation,  the plane-wave state is an eigenvector with an eigenphase proportional to $f(x)$. Ref. \onlinecite{jordan2005fast} sets the actual value added to the second register as $\frac{MN}{WL} f(x)$ rounded to the nearest integer. Subsequently, phase kickback to the $\epsilon$ register gives:
\begin{align}\label{QGAphase_apx}
& \left( \frac{1}{\sqrt{M}}\sum_{\epsilon} e^{i2\pi\frac{M}{WL}f\left(\frac{L}{M}\left(\epsilon-\frac{M}{2}\right)\right)} |\epsilon\rangle \right)
 \left(\frac{1}{\sqrt{N}} \sum_a e^{i2\pi a/N} |a\rangle\right) 
 \end{align}
where $W$ is the maximum expected value for the gradient, and $L$ is a parameter representing a small linear region for the domain of which the function is being evaluated.

If we choose $L$ sufficiently small over a region where $f$ is linear, then we can approximate the function $f$ with a series expansion to first order as
\begin{equation}
f\left(\frac LM\left(\epsilon-\frac{M}2\right)\right)\approx f({0})+\frac LM\left(\epsilon-\frac{{M}}2\right)\cdot{\nabla}f
\end{equation}

The leading term $f(0)$ contributes only a global phase so it can safely be ignored since it is irrelevant on measurement. The terms that depend on $\mathbf{\epsilon}$ can be isolated, so that the state of the first register simplifies to: 
\begin{equation}
\frac{1}{\sqrt{M}}\sum_{\epsilon=0}^{M-1}e^{i\frac{2\pi }{W}\epsilon\cdot{\nabla}f}|{\epsilon}\rangle
\end{equation}

The final step of the QGA is to apply a QFT on the $m$-qubit register giving:
\begin{equation}\label{gradientoutcome}
\left|\frac MW\frac{\partial f}{\partial x}\right\rangle
\end{equation}
which yields the derivative of the function with respect to $x$ written out in the computational basis. Note that this method only requires one call to the oracle and is a notable speedup over the classical. The QFT can be implemented with $O(m^2)$ gates \cite{nielsen2010quantum,coppersmith2002approximate}.

\section{Implementation of the quantum gradient phase estimation}
\label{qgaimplement}

In the following, we describe the approach for determining the gradient of an eigenvalue with respect to changes in the elements of a matrix. We perturb the original matrix  $\hat H_0$ so that the new matrix is defined by
\begin{equation}
\hat{H}(\epsilon) = \hat H_0 + \frac{L\epsilon}{N} \Delta(a,b)
\end{equation}
where $N = 2^n$ is the size of the Hilbert space of an $n$-qubit register, $L$ is a small quantity representing a linear region for the gradient of the perturbation, $\epsilon \in \{0,N-1\}$ an integer and $\Delta_{(ab)}$ is a sparse Hermitian matrix with one in the matrix position being perturbed corresponding to the element in row $a$ and column $b$. To ensure Hermitianity,  if $a\neq b$ then necessarily the element in row $b$ and column $a$ must also be set to 1.  Alternatively (as in Ref.~\onlinecite{jordan2005fast}), we can have that
\begin{equation}
\hat{H} (\epsilon) =\hat H_0 + \frac{L}{N} \left(\epsilon - \frac{M}{2}\right)\Delta(a,b)
\end{equation}
which is a shifted form of what was used in the text. In the above formulation, the matrix elements are deviated by values within the range $[-L/2, +L/2]$  in place of $[0,+L]$ as in the previous. 

\subsection{Oracle query for the gradient}

 The first register is the answer register which is encoded onto $n$ qubits and the second register is encoded with an eigenstate of the original matrix. As in quantum phase estimation, the general sequence of the circuit is to first perform a Hadamard transform on the first register, followed by a controlled unitary operation which targets the second register, and finally the inverse quantum Fourier transform on the first register.  The initial Hadamard ($H$) transform on the first register forms an equal superposition of input basis states $\{\ket{\epsilon}\}$. So, we have
\begin{equation}
\ket{0}^{\otimes m} \xrightarrow{H^{\otimes m}} \frac{1}{\sqrt{M}} \sum_{\epsilon =0}^{M-1}\ket{\epsilon} \end{equation}
which effectively maps the $0$ state initialized onto $n$ qubits into a register of all possible combinations of numbers. This accounts for all possible directions that the gradient can be taken.

On the second register, there is the eigenstate preparation for $\ket{p}$ using the $\Gamma^{(p)}$ operator
\begin{equation}
\ket{0}^{\otimes n} \xrightarrow{\Gamma^{(p)}}\ket{p}
\end{equation}
where $p$ denotes the eigenstate of the matrix $\hat H_0$ that is being prepared.

Next, there is a unitary operation that is controlled ($\mathcal{U}_c$) by the input register which targets the eigenstate register. This operator has the format of
\begin{equation}\label{Uc}
\mathcal{U}_c =  \sum_{\epsilon =0}^{M-1}\ketbra{\epsilon}{\epsilon} \otimes \hat{U}(\epsilon)
\end{equation}
where the unitary being applied is given as Eq.~\eqref{eq:u_ep}. The value  $W$ is some maximum anticipated value of the gradient. The form of $\mathcal{U}(\epsilon)$ is given by
\begin{equation}
\mathcal{U}(\epsilon) = e^{i\frac N{WL}\hat{H}(\epsilon)} \label{eq:u_ep}
\end{equation}
where $N$, $W$, and $L$ were defined the same as in the main text.

If we assume initially that $[H_0, \Delta_{(ab)}] = 0$, the eigenvalue equation is expressed here as
\begin{equation}
\hat H_0 \ket{\Psi} = \lambda_0 \ket{\Psi} 
\end{equation}
and then when the operator is controlled this extended to first order as
\begin{align}
\hat{H}(\epsilon) \ket{\Psi} &= \Big(\lambda_0 + \delta \lambda (\epsilon)\Big) \ket{\Psi} \\
\hat{H}(\epsilon) \ket{\Psi}  & = \lambda (\epsilon) \ket{\Psi} \numberthis
\end{align} 
where each wavefunction is not controlled on $\epsilon$.

The eigenvalue corresponding to $\hat U(\epsilon)$ is as in Eq.~\eqref{eq:u_eig} since again we assume here that $\hat H_0$ and $\Delta_{(ab)}$ is commuting. Therefore,
\begin{align}
e^{i\frac N{WL}\lambda(\epsilon)} &= e^{i\frac N{WL} [\lambda_0 + \delta \lambda(\epsilon)]} = e^{i\frac N{WL}\lambda_0} \left(e^{i\frac1W  \epsilon \nabla_{(ab)}\lambda}\right) 
\label{eq:u_eig} 
\end{align}
with  $\nabla_{(ab)}\lambda$ as the gradient of $\lambda$ with respect to infinitesimal changes to the entries of the matrix indicated by the non-zero elements of $\Delta_{(ab)}$, for which we have rewritten 
\begin{equation}
\delta \lambda(\epsilon) = \frac{L\epsilon}{N} \nabla_{(ab)}\lambda\label{eq:delta_gradl}
\end{equation}
giving now the full equation.

 The last step is that the inverse quantum Fourier transform ensures that the gradient $\nabla_{(ab)} \lambda$ is written into register basis, since the pre-factor containing $\lambda_0$ is a global phase.
 
 This procedure is then the set-up for the oracle query of the QGA.

\subsection{Implementation of the quantum Fourier transform}

The entire circuit evolves as outlined in the following. Equations~(\ref{eq:1}) - (\ref{eq:3}) show the transformation with the controlled unitary operation on the two registers, noting that $\lambda$ is replaced by $2\pi \theta$ for convenience with QFT convention in the subsequent step. This is written explicitly as
\begin{equation}
\left(\frac{1}{\sqrt{M}} \sum_{\epsilon =0}^{M-1}\ket{\epsilon} \right) \otimes \ket{p} \label{eq:1} \end{equation}
and then followed by application of $\mathcal{U}_c$ from earlier to give
\begin{align} 
&\xrightarrow{\mathcal{U}_c}  \frac{e^{i2\pi \frac N{WL}\theta^{(p)}_0 }}{\sqrt{M}}\sum_{\epsilon =0}^{M-1} \ket{\epsilon}\ket{p(\epsilon)} \label{eq:2} \\
&= \frac{e^{i2\pi  \frac N{WL}\theta^{(p)}_0}}{\sqrt{M}}\left(\sum_{\epsilon =0}^{M-1} e^{i2\pi \frac1W \epsilon  \nabla_{(ab)}\theta^{(p)} 
}\ket{\epsilon}\ket{p}\right) \label{eq:3} 
\end{align}
which is now the eigenvalue for each possible $\varepsilon$.

The inverse QFT operator on the first register alone is defined in Eq.~\eqref{eq:4}. This acts on the state at Eq.~\eqref{eq:3} to give Eq.~\eqref{eq:5}. This is expressed by $\mathbb{I}$ is the identity operator)
\begin{equation}
\QFT^{-1} = \frac{1}{\sqrt{M}}\sum_{j,k=0}^{M-1} e^{-i2\pi jk/M}\ketbra{j}{k} \otimes \mathbb{I} \label{eq:4}
\end{equation}
and then application onto the existing circuit gives
\begin{equation} 
 \xrightarrow{\QFT^{-1}} \frac{\alpha}{M}\sum_{j,k=0}^{M-1} e^{-i2\pi k\left(j-\frac NW\nabla_{(ab)}\theta^{(p)}\right)/M} \ket{j}\otimes\ket{p} \label{eq:5}
 \end{equation}
which is now the full Fourier transform of the input, changing $\epsilon$ to a number $j$. We have denoted
\begin{equation}
\alpha = e^{i2\pi \frac N{WL}\theta^{(p)}_0}
\end{equation}
as the global phase. The wavefunction peaks at the value of $j$ in the first register such that 

\begin{align} 
j &= \frac{N}{W} \nabla_{(ab)}\theta^{(p)} = \frac{N}{2\pi W}\nabla_{(ab)}\lambda^{(p)} \end{align}
and therefore is the full expression of what is communicated in Eq.~\eqref{gradientoutcome} under the conventions chosen here.

\subsection{Example 1: Pauli X Gate}

As the first example, we consider the matrix for the $X$-gate \cite{townsend2000modern}
 \begin{equation}
 X = \begin{pmatrix}0 & 1 \\ 1 & 0 \end{pmatrix}
 \end{equation}
 which is useful since it is real and displays off-diagonal elements \cite{nielsen2010quantum}.

An eigenstate of  $X$ may be encoded onto a single qubit. For simple evaluation, we consider a single qubit for the input register as well so that $\epsilon \in \{0,1\}$. The total number of qubits required for this gradient estimation circuit is only $2$. This circuit evolves as
\begin{equation}
\frac{1}{\sqrt{2}} \left(\ket{0} + \ket{1}\right) \otimes \ket{p}
\end{equation}
where we now apply $\mathcal{U}_c$ just as above to obtain
\begin{align} \xrightarrow{\mathcal{U}_c}& \frac{1}{\sqrt{2}}\left(e^{i\frac{4\pi}L\theta^{(p)}_0}\ket{0}\ket{p} + \ket{1}\ket{p(\epsilon)}\right) \\
=& \frac{1}{\sqrt{2}}\left(e^{i\frac{4\pi}L\theta^{(p)}_0}\ket{0}\ket{p} + e^{i\frac{4\pi}L\left(\theta_0^{(p)} +\nabla_{(ab)}\theta^{(p)}\frac L2\right)}\ket{1}\ket{p}\right) \\
 =& \frac{e^{i\frac{4\pi}L\theta^{(p)}_0}}{\sqrt{2}} \left(\ket{0}\ket{p} + e^{i2\pi\nabla_{(ab)}\theta^{(p)}}\ket{1}\ket{p}\right) \\
=& \frac{e^{i\frac{4\pi}L\theta^{(p)}_0}}{\sqrt{2}}  \sum_{\epsilon=0}^{1} e^{i2\pi \epsilon \nabla_{(ab)}\theta^{(p)}}\ket{\epsilon}\ket{p} 
\end{align}
so that the circuit is now ready for a phase kick-back. In order to bring the phase onto a register for measurement, we have
\begin{align}
\xrightarrow{\QFT^\dagger}& \left(\frac{1}{\sqrt{2}}\sum_{j=0}^{1}\sum_{k=0}^{1} e^{-i2\pi jk/2}\ketbra{j}{k} \otimes \mathbb{I}\right)\frac{e^{i\frac{4\pi}L\theta^{(p)}_0}}{\sqrt{2}}  \sum_{\epsilon=0}^{1} e^{i2\pi \epsilon\nabla_{(ab)}\theta^{(p)}}\ket{\epsilon}\ket{p} \\
 = &\frac{e^{i\frac{4\pi}L\theta^{(p)}_0}}{2}\sum_{j=0}^{1} \sum_{k=0}^{1}e^{-i2\pi k(j-2\nabla_{(ab)}\theta^{(p)})/2} \ket{j}\otimes\ket{p} \label{eq:2qbit} 
 \end{align}
and is now displayed in this form to show that the most likely state will collapse to 
\begin{align} 
j &= 2\nabla_{(ab)}\theta^{(p)} = \frac1\pi\nabla_{(ab)}\lambda^{(p)} 
\end{align}
where we have chosen $m=1$ for this problem. 

Starting from the expression in Eq.~\eqref{eq:2qbit}, the exact gradient may be deduced according to the output wavefunction amplitudes. We let $c_0$ and $c_1$ be the output amplitudes corresponding to outcome basis states $\ket{j=0}$ and $\ket{j=1}$ respectively. Calling the global phase $\alpha =  e^{i\frac{4\pi}L\theta^{(p)}_0} $ and writing $2\pi \nabla_{(ab)}\theta^{(p)}=\nabla_{(ab)}\lambda^{(p)}$ , we get
\begin{align}
c_0 &= \frac{e^{i\frac{4\pi}L\theta^{(p)}_0}}{2} \left( 1 + e^{i2\pi \nabla_{(ab)}\theta^{(p)}}\right) = \frac{\alpha}{2}\left(1 + e^{i\nabla_{(ab)}\lambda^{(p)}} \right) \numberthis 
\end{align}
and
\begin{align}  c_1 &= \frac{e^{i\frac{4\pi}L\theta^{(p)}_0}}{2} \left( 1 - e^{i2\pi\nabla_{(ab)}\theta^{(p)}}\right)= \frac{\alpha}{2}\left(1 - e^{i\nabla_{(ab)}\lambda^{(p)}} \right) \numberthis 
 \end{align}
and we know that  $0\leq|c_0|^2\leq 1$ and $|\alpha|^2 = 1$. So, by taking the squared modulus, we deduce the following for our single qubit case, we have that
\begin{equation}
|c_0|^2 = \left|\frac{\alpha}{2}\right|^2 \left| 1 + e^{i\nabla_{(ab)}\lambda^{(p)}}\right|^2
\end{equation}
and therefore
\begin{equation} 
 \nabla_{(ab)}\lambda^{(p)} = 
 \begin{cases} 
  \cos^{-1}(2|c_0|^2 - 1) \equiv 2\cos^{-1}\left(\sqrt{|c_0|^2}\right)   \\
  \\
  \cos^{-1}\left(1-2|c_1|^2\right) \equiv 2\sin^{-1}\left(\sqrt{|c_1|^2}\right) \\
   
 \end{cases}
\end{equation}
gives the full output of the circuit measurement.

Since the answer register is coupled to the eigenstate $\ket{p}$,  we undo the eigenstate preparation operator at the end of the circuit in order to facilitate the analysis of the statevector output for the circuit. 

Implementations of this algorithm to different cases will inevitably require different architectures, especially considering that we do not consider the quantum computer's architecture here. We leave further development on this to future works. This analysis here can be useful for calibrating any such technique and checking on a small problem. We also point to the other examples on a $Z$-gate and Hadamard operator for further instances, all of which follow trivially from this implementation with results shown in Table~\ref{outputs}.

\begin{table}[]
\begin{tabular}{|c|c|c|}
\hline
& & \\ 
\; $\Delta(a,b)$ \; &\; $\ket{\Psi ^{(p)}}$ \;& \;$\nabla_{(ab)}\lambda^{(p)}$ \; \\ 
 & & \\  
\hline
\hline
 & & \\
    \; $X = \begin{pmatrix} 0 & 1\\ 1 & 0 \end{pmatrix}$   \;    &    \;    $\ket{+}$  \;     &      0.999999             \\ 
         &       \; $\ket{-} $   \;                &      0.999999                \\ 
         & & \\ 
         \hline
         & & \\
    \; $ \ketbra{0}{0} = \begin{pmatrix} 1 & 0 \\ 0 & 0 \end{pmatrix}$ \;  & \;    $\ket{+}$  \;  & 0.500000   \\
     & \;    $\ket{-}$  \;  & 0.499999  \\ 
     & & \\ \hline
     & & \\
   \; $ \ketbra{1}{1} = \begin{pmatrix} 0 & 0 \\ 0 & 1 \end{pmatrix}$ \;              &      \;    $\ket{+}$  \;                &  0.500000                   \\
   & \;    $\ket{-}$  \;  & 0.499999 \\ 
  & & \\ \hline
               & & \\
  $I = \begin{pmatrix} 1 & 0 \\ 0 &1 \end{pmatrix}$   \;               &         \;    $\ket{+}$  \;            &         0.999999           \\
   &       \; $\ket{-} $   \;                &      1.000000             \\ 
  & & \\ \hline
\end{tabular}
\caption{Gradient outputs using DMRjulia for $\hat H_0 = X$ and $L= 10^{-6}$. 
\label{outputs}
}
\end{table}

\subsection{Results}

We choose $L= 10^{-6}$ and encode either of the eigenstates of $X$ which are the $\ket{+} = \frac{1}{\sqrt{2}}\left(\ket{0}+\ket{1}\right)$ and $\ket{-} =  \frac{1}{\sqrt{2}}\left(\ket{0}-\ket{1}\right)$. The state $\ket{+}$ is prepared using the a Hadamard gate and the state $\ket{-}$  by applying an $X$-gate before a Hadamard gate. We also test for different forms of the perturbation matrix $\Delta(a,b)$, all of which match the expected result.

\subsection{Example 2: Hadamard Gate}

Considering now the Hadamard gate as the unperturbed matrix $H_0$,
\begin{equation}
\ket{H_+} =  \cos\left(\frac{\pi}{8}\right)\ket{0} + \sin\left(\frac{\pi}{8}\right)\ket{1}
\end{equation}
and
\begin{equation}
\ket{H_-} = -\sin\left(\frac{\pi}{8}\right)\ket{0} + \cos\left(\frac{\pi}{8}\right)\ket{1}
\end{equation}
as eigenvectors.

The state $\ket{H_+}$ is prepared using the rotational $R_y(\frac{\pi}{4})$ gate and
the state $\ket{H_-}$  by applying an $X$-gate before an $R_y(\frac{\pi}{4})$ gate.
Again, $L=1\times 10^{-6}$ and we choose $\Delta(a,b) = X$. For both $\ket{H_+}$ and  $\ket{H_-}$, we get an output of about $\nabla_{ab}\lambda^{(p)} = 0.70710691$ which is approximately equal to $1/\sqrt{2}$ with a difference of $\sim10^{-7}$. Better approximations are achieved by further decreasing the value of the parameter $L$, or extrapolating. We note the error here decreases linearly for the idealized case. Results on different quantum computers may vary.

\subsection{Simulation details}

We built a quantum simulator on top of a matrix product state in the DMRjulia library \cite{bakerCJP21,*baker2019m,dmrjulia0,dmrjulia,dmrjulia1}. All basic algorithms were tested to ensure accuracy. Implementation details of this library will be detailed later, but the simulator is compiler-assumption-free.

\section{Multi-variable series expansion}\label{multiTaylor}

The series expansion over several variables takes the form
\begin{equation}
f(\mathbf{x})=\sum_{n_1=0}^\infty\ldots\sum_{n_N=0}^\infty\frac{(x_1-a_1)^{n_1}\ldots(x_N-a_N)^{n_N}}{n_1!\ldots n_N!}\frac{\partial^{n_1+\ldots+n_N}f}{\partial x_1^{n_1}\ldots\partial x_N^{n_N}}(a_1,\ldots,a_N)
\end{equation}
where $\mathbf{x}=\langle x_1,x_2,x_3,x_4,x_5,\ldots x_N\rangle$ and the derivative is centered at $\mathbf{a}=\langle a_1,a_2,a_3,\ldots,a_N\rangle$.


\end{appendix}

\end{document}